\documentclass[a4paper,UKenglish,cleveref, autoref, thm-restate]{lipics-v2021}

\pdfoutput=1 
\hideLIPIcs  

\title{Quantum Byzantine Agreement Against Full-information Adversary} 

\author{Longcheng Li}{State Key Lab of Processors, Institute of Computing Technology, Chinese Academy of Sciences, Beijing 100190, China \and School of Computer Science and Technology, University of Chinese Academy of Sciences, Beijing 100049, China}{lilongcheng22s@ict.ac.cn}{https://orcid.org/0000-0002-5259-9807}{}

\author{Xiaoming Sun}{State Key Lab of Processors, Institute of Computing Technology, Chinese Academy of Sciences, Beijing 100190, China \and School of Computer Science and Technology, University of Chinese Academy of Sciences, Beijing 100049, China}{sunxiaoming@ict.ac.cn}{https://orcid.org/0000-0002-0281-1670}{}

\author{Jiadong Zhu\footnotemark[1]}{State Key Lab of Processors, Institute of Computing Technology, Chinese Academy of Sciences, Beijing 100190, China}{zhujiadong2016@163.com}{https://orcid.org/0000-0003-4701-9967}{}

\authorrunning{L. Li, X. Sun, and J. Zhu} 

\Copyright{Longcheng Li, Xiaoming Sun, and Jiadong Zhu} 



\category{} 

\relatedversion{} 

\funding{This work was supported in part by the National Natural Science Foundation of China Grants No.~62325210, and the Strategic Priority Research Program of Chinese Academy of Sciences Grant No.~XDB28000000.}

\nolinenumbers 

\EventEditors{Dan Alistarh}
\EventNoEds{1}
\EventLongTitle{38th International Symposium on Distributed Computing (DISC 2024)}
\EventShortTitle{DISC 2024}
\EventAcronym{DISC}
\EventYear{2024}
\EventDate{October 28--November 1, 2024}
\EventLocation{Madrid, Spain}
\EventLogo{}
\SeriesVolume{319}
\ArticleNo{25}
\usepackage{braket}
\usepackage{mdframed}
\usepackage{algorithm}
\usepackage[noend]{algpseudocode}
\usepackage{booktabs}
\usepackage{multirow}
\usepackage[normalem]{ulem} 

\usepackage{thmtools,thm-restate}

\algdef{SL}[REPLICASTATE]{ReplicaState}{0}[1]{\textbf{state} #1}
\algblockdefx[FN]{Fn}{EndFn}%
[3]{\textbf{function} \emph{#1}(#2) : #3}{}
\algblockdefx[PROC]{Proc}{EndProc}%
[2]{\textbf{function} \emph{#1}(#2)}{}
\algblockdefx[UPON]{Upon}{EndUpon}%
[1]{\textbf{upon} #1}{}
\algtext*{EndUpon}
\algtext*{EndFn}
\algtext*{EndProc}

\newcommand{\View}{\mathsf{View}}
\newcommand{\Viewrm}{\textrm{View}}

\newcommand{\Msf}{\mathsf{M}}
\newcommand{\Asf}{\mathsf{A}}
\newcommand{\Rsf}{\mathsf{R}}
\newcommand{\Dsf}{\mathsf{D}}
\newcommand{\Bsf}{\mathsf{B}}
\newcommand{\Ssf}{\mathsf{S}}
\newcommand{\Qsf}{\mathsf{Q}}

\newcommand{\PC}{\mathcal{P}_{C}}
\newcommand{\PQ}{\mathcal{P}_{Q}}
\newcommand{\AC}{\mathcal{A}_{C}}
\newcommand{\AQ}{\mathcal{A}_{Q}}

\newcommand{\Mcal}{\mathcal{M}}

\renewcommand{\star}{*}

\newcommand{\E}{\mathop{\mathbb{E}}}

\begin{document}

\maketitle

\footnotetext[1]{Corresponding author}
\begin{abstract}
We exhibit that, when given a classical Byzantine agreement protocol designed in the private-channel model, it is feasible to construct a quantum agreement protocol that can effectively handle a full-information adversary. Notably, both protocols have equivalent levels of resilience, round complexity, and communication complexity. In the classical private-channel scenario, participating players are limited to exchanging classical bits, with the adversary lacking knowledge of the exchanged messages. In contrast, in the quantum full-information setting, participating players can exchange qubits, while the adversary possesses comprehensive and accurate visibility into the system's state and messages. By showcasing the reduction from quantum to classical frameworks, this paper demonstrates the strength and flexibility of quantum protocols in addressing security challenges posed by adversaries with increased visibility. It underscores the potential of leveraging quantum principles to improve security measures without compromising on efficiency or resilience.

By applying our reduction, we demonstrate quantum advantages in the round complexity of asynchronous Byzantine agreement protocols in the full-information model. It is well known that in the full-information model, any classical protocol requires $\Omega(n)$ rounds to solve Byzantine agreement with probability one even against Fail-stop adversary when resilience $t=\Theta(n)$ \cite{attiya2008tight}. We show that quantum protocols can achieve $O(1)$ rounds (i) with resilience $t<n/2$ against a Fail-stop adversary, and (ii) with resilience $t<n/(3+\epsilon)$ against a Byzantine adversary for any constant $\epsilon>0$, therefore surpassing the classical lower bound.

\end{abstract}

\section{Introduction}

Byzantine agreement (BA) \cite{10.1145/322186.322188}, also referred to as Byzantine fault-tolerant distributed consensus, is a crucial topic in secure distributed computing. In simple terms, in a BA protocol, a group of $n$ players who do not trust each other and possess private input bits, come to a consensus on a shared output bit, even if a subset of size $t$ of the players are corrupted by a malicious adversary, who can force the corrupt parties to deviate from their prescribed programs during the protocol execution. The Byzantine agreement problem has been extensively researched over the past four decades, leading to numerous findings on the feasibility and potential of BA protocols in various settings \cite{doi:10.1137/S0097539790187084,10.1145/31846.42229,10.5555/983102}.

In this paper, we focus on BA that succeeds with probability one
in the full-information model, where the adversary knows the knowledge of all local variables, including quantum states if applicable. It is well
known that in this model, when up to $t$ players may be corrupted, no classical deterministic
protocol can solve synchronous BA in less than $t+1$ rounds even in the presence
of a Fail-stop adversary \cite{10.1145/322186.322188}. It is further proved by
\cite{10.1145/277697.277733} that any classical randomized protocol requires at
least expected $\tilde{\Omega}(\sqrt{n})$ rounds. Given these constraints, it is natural to ask the
following question:
\begin{center}
    \emph{
    Can quantum communication accelerate BA in the full-information model?
    }
\end{center}

The seminal work of \cite{10.1145/1060590.1060662} provides a confirming answer to the above question by constructing a constant round synchronous quantum BA
protocol against the Byzantine adversary,  surpassing the established round complexity lower bound in \cite{10.1145/277697.277733}. The protocol builds upon an expected constant round
classical BA protocol introduced in \cite{doi:10.1137/S0097539790187084}, which is not
resilient against a full-information adversary and requires a private channel.
In their work, \cite{10.1145/1060590.1060662} proposes a quantum modification to
the original classical protocol to make it robust against a full-information
adversary. They achieve this by introducing a novel approach of deferring coin
flips, substituting them with quantum superpositions until after the adversary
has chosen his actions in a certain round. Notably, the modification does not
change the structure of the original classical protocol and therefore preserves
its constant round complexity. \cite{10.1145/1060590.1060662} also extends the synchronous quantum protocol to the
asynchronous case, but with suboptimal resilience $t<n/4$. 


\cite{10.1145/1060590.1060662} demonstrates an elegant method of reducing quantum full-information protocols to classical private-channel protocols while maintaining key attributes such as resilience, round complexity, and communication complexity. This approach offers a valuable means of evaluating the quantum advantage in the full-information model by comparing classical full-information and classical private-channel models. By highlighting the notable distinctions between these two classical models, they underscore the substantial quantum advantage inherent in the full-information domain. 

In light of these findings, \cite{10.1145/1060590.1060662} raises the question of whether their reduction strategy could be applied to other settings, such as low round complexity asynchronous BA protocols with resilience $n/4 \leq t<n/3$, to further investigate potential quantum advantages. Unfortunately, limited progress has been made on this issue since its introduction. This paper seeks to tackle this challenge from a comprehensive viewpoint. Instead of narrowly focusing on the reduction of quantum protocols to classical protocols in a specific setting (e.g., the asynchronous protocol with resilience $n/4 \leq t<n/3$, which is better than that in \cite{10.1145/1060590.1060662}), our objective is to address the following question:
\begin{center}
    \emph{Is it possible to convert \textbf{any} classical private-channel BA protocol to a quantum full-information BA protocol  while preserving the same characteristics such as\\ resilience, round complexity, and communication complexity?}
\end{center}


\subsection{Our Contribution}

As our main result, we answer the above question in the affirmative by demonstrating a general reduction from a quantum
full-information BA protocol to a classical private-channel BA protocol:

\begin{restatable}{thm}{maintheorem}
\label{thm:main}
    Given a classical synchronous (resp.~asynchronous) non-erasing BA protocol designed to counter a private-channel Fail-stop (resp.~Byzantine) adversary, we can construct a quantum synchronous (resp.~asynchronous) BA protocol capable of handling a full-information Fail-stop (resp.~Byzantine) adversary while maintaining the same levels of resilience, round complexity, and communication complexity.
\end{restatable}



It is crucial to emphasize that the theorem we present is applicable under the condition that the classical protocol forming the foundation of our quantum protocol is \emph{non-erasing}, which means its
security does not rely on the erasure of intermediate states. 
To the best of our knowledge, this criterion is met by all existing classical protocols within
the scope of information-theoretic BA with probability one. 
For a more detailed definition of this concept, please refer to \Cref{def:nonerase} in Section \ref{sec:main} where we will provide a formal explanation. Furthermore, throughout our paper, we consistently assume that the adversary is computationally unlimited and adaptive\footnote{Similar reductions can also be made from the quantum non-adaptive full-information model to the classical non-adaptive private-channel model.}, allowing it to modify its strategy based on the information acquired during the execution of the protocol.


We remark that our reduction does not preserve the local computational
time/space and expected number of communication bits of the original protocol,
as we assume players have unbounded local computational power and remember all
intermediate states. However, for certain classical BA protocols, with a more careful design of the reduction strategy,
it is possible to avoid the blow-up of local computation and the number of
communication bits in the corresponding quantum protocols. For details, we refer interested readers to \Cref{sec:bit-efficient} for more information.

\begin{table}
    \centering
    \begin{threeparttable}
    \caption{Round complexity of Byzantine agreement in the full-information model. 
    }
    \label{tab:ourresult}
    \begin{tabular}{cccccccc}
        \toprule[1.5pt]
        \multirow{2}{*}{Model} & \multirow{2}{*}{Adversary} & \multirow{2}{*}{Resilience} & \multicolumn{2}{c}{Classical}  & ~ &  \multicolumn{2}{c}{Quantum\tnote{a}}\\
        \cline{4-5} \cline{7-8} 
        &  &  & Upper bound & Lower bound & ~ & Upper bound & $\PC$ \\
        \midrule[1.5pt]
        \multirow{2}{*}{Sync.} & Fail-stop & \(t=\Theta(n)\) & \(\tilde{O}(\sqrt{n})\) \cite{10.1145/277697.277733} & \multirow{2}{*}{\(\tilde{\Omega}(\sqrt{n})\) \cite{10.1145/277697.277733}} & ~ & \(O(1)\) \cite{10.1145/1060590.1060662,hajiaghayi2023brief} & \cite{chor1989simple}\\
        & Byzantine & \(t<n/3\) & \(O(n)\) \cite{10.1145/322186.322188} & & ~ & \(O(1)\) \cite{10.1145/1060590.1060662} & \cite{doi:10.1137/S0097539790187084}\\
        \midrule
        \multirow{4}{*}{Async.} & Fail-stop & \(t=\Theta(n)\) & \(O(n)\) \cite{attiya2008tight} & \multirow{4}{*}{\(\Omega(n)\) \cite{attiya2008tight}} & ~ & \textcolor{red}{\(O(1)\) (Our work)} & \cite{10.5555/983102}\\
        & Byzantine & \(t<n/4\) & \(\tilde{O}(n^4)\) \cite{huang2023byzantine} & & ~ & \(O(1)\) \cite{10.1145/1060590.1060662} & \cite{doi:10.1137/S0097539790187084}\\
        & Byzantine & \(t<\frac{n}{3+\epsilon}\)\tnote{b} & \(\tilde{O}(n^4/\epsilon^8)\) \cite{huang2023byzantine} &  & ~ & \textcolor{red}{\(O(1/\epsilon)\) (Our work)} & \cite{bangalore2018almost}\\
        & Byzantine & \(t<n/3\) & \(\tilde{O}(n^{12})\) \cite{huang2023byzantine} &  & ~ & \textcolor{red}{\(O(n)\) (Our work)} & \cite{bangalore2018almost}\\
        \bottomrule[1.5pt]
    \end{tabular} 
    \begin{tablenotes}
       \item[a] Every quantum protocol presented in the table is built upon some classical private-channel protocol $\PC$. The last two columns of the table show the classical private-channel protocols alongside their quantum full-information equivalents for comparison and reference purposes.
       \item[b] 
       Notice that when $\epsilon$ is a constant, the quantum upper bound is $O(1/\epsilon)=O(1)$.
    \end{tablenotes}    
    \end{threeparttable}
\end{table}

By applying our reduction, we obtain several new quantum advantages related to round complexity in the full-information setting. As summarized in Table \ref{tab:ourresult}, our main result enables us to
quantize existing classical private-channel protocols into some quantum
full-information protocols of which the round complexity surpasses the classical
lower bound in the same setting. In particular, we obtain two new quantum
speedups in the asynchronous model:
\begin{itemize}
\item {\bf Fail-stop model:} Section 14.3 of \cite{10.5555/983102} presents a
constant-round classical BA protocol with optimal resilience $t<n/2$ against the
Fail-stop adversary in the private-channel setting. By applying our reduction, we obtain a
constant-round quantum full-information BA protocol with $t<n/2$, while any
classical full-information protocol requires $\Omega(n)$ rounds \cite{attiya2008tight}.
\item {\bf Byzantine model:} 
For any $\epsilon>0$, \cite{bangalore2018almost} presents an $O(1/\epsilon)$-round classical BA
protocol with resilience $t<n/(3+\epsilon)$ against the private-channel
Byzantine adversary. By applying our reduction, we
obtain an $O(1/\epsilon)$-round quantum full-information BA protocol with
resilience $t<n/(3+\epsilon)$. When $\epsilon$ is a constant independent of $n$,
the quantum BA achieves constant rounds, while any classical full-information
protocol requires $\Omega(n)$ rounds \cite{attiya2008tight}. When $\epsilon \leq 1/n$, $\lceil n/(3+\epsilon) \rceil=\lceil n/3 \rceil$, which indicates that $t<n/(3+\epsilon)$ is equivalent to $t<n/3$. By substituting $\epsilon = 1/n$ into $O(1/\epsilon)$, we find that our
quantum BA requires $O(1/\epsilon)=O(n)$ rounds. In comparison, the
best known classical protocol \cite{huang2023byzantine} in the same setting
requires $\tilde{O}(n^{12})$ rounds.
\end{itemize}

\subsection{Technical Overview}\label{subsec:Tech}

We briefly explain the key ideas behind \Cref{thm:main}, especially how to
quantize a classical protocol into a quantum one and how to simulate a quantum
full-information adversary in the classical setting. The key idea is utilizing
quantum superpositions to turn exposed randomness into hidden randomness. 

\subparagraph*{A simple motivating example.}
Before introducing the complicated quantum full-information BA protocol against the Byzantine adversary, \cite{10.1145/1060590.1060662} first presents a simple quantum full-information BA protocol against the Fail-stop adversary, who can corrupt players by halting it and choosing a subset of their messages
to be delivered. This simple protocol follows a common framework of reducing a BA protocol to a common-coin protocol, where all uncorrupted players need to output a common
random coin with constant success probability. We will use the common-coin protocol, as demonstrated in the BA protocol against the Fail-stop adversary in \cite{10.1145/1060590.1060662}, as a motivating example to explain the key idea of our paper. The common-coin protocol works in the quantum full-information setting and draws inspiration from a common-coin protocol in the classical
private-channel setting \cite{chor1989simple}.
In the following discussion, we will start by offering a brief overview of the classical private-channel protocol in \cite{chor1989simple} and explaining its limitations when confronted with a full-information adversary. We then explain how \cite{10.1145/1060590.1060662} effectively resolve this issue by leveraging quantum principles. 

The classical
private-channel protocol in \cite{chor1989simple} works as follows: (i) Each player $i$ picks a random
coin $c_i\in\{0,1\}$ and a random leader value $l_i\in [n^3]$ and then
multicasts $(c_i, l_i)$; (ii) Each player $i$ outputs the coin $c_j$ such that
$l_j$ is the largest leader value $i$ receives. A private-channel Fail-stop
adversary learns nothing about the values of $\{c_i\}$ and $\{l_i\}$, so the
best it can do is to randomly stop $t$ players. Since there are at least $n-t>n/2$
uncorrupted players, the largest leader falls among uncorrupted players with probability~$1/2$, and the
probability of collision of leader values is negligible. Switching to full-information adversary, $\{c_i\}$ and $\{l_i\}$ become known to
the adversary. Then the adversary can corrupt the leader and let only a subset
of players receive the leader's message so that it can break the common-coin
protocol. However, \cite{10.1145/1060590.1060662} shows that the problem can be
fixed if we allow quantumness. Instead of choosing random $c_i$ and $l_i$, we
let player $i$ \emph{purify randomness}, i.e, preparing two $n$-qudit
superposition states
\[\ket{c_i}:=\frac{1}{\sqrt{2}}\left(\ket{00\cdots 0}+\ket{11\cdots 1}\right) \text{ and }
\ket{l_i}:=\frac{1}{\sqrt{n^3}}\sum_{l=1}^{n^3}\ket{l,l,\ldots,l},\]
and then distribute the $n$ qudits of $\ket{c_i}$ and $\ket{l_i}$ among the
players. In the next round, the players measure the qudits they receive and obtain
the classical random coins and leader values. Although the full-information adversary 
can see the pure state of the system, quantum mechanics
prevents it from knowing the random values before measurement. Thus this simple
purified quantum protocol works against the full-information adversary.


\subparagraph*{Generalized reduction in the synchronous model.} Inspired by the
above example, we give a general reduction from quantum full-information BA protocols to classical private-channel
BA protocols. For any
classical BA protocol $\PC$, the local computation of each player at
round $k$ involves (i) preparing some randomness $r_k$,  and (ii) computing a function
$f$ to determine the decided value and messages to be sent. We construct a
quantum protocol $\PQ$ by modifying $\PC$'s local computation to (i) preparing a
quantum state $\sum_{r}\sqrt{\Pr[r_k=r]}\ket{r}$, (ii) applying a unitary $U_f$
to compute $f$ reversibly i.e., $U_f\ket{v}\ket{0}:=\ket{v}{\ket{f(v)}}$ and
send quantum messages. 


We assume that the output of $f$
contains a variable $d_k\in\{0, 1, \perp\}$ indicating the decided value at
round $k$ ($\perp$ if not decided yet). The player in $\PQ$ will
measure the corresponding quantum register of $d_k$ and decide if
$d_k\not=\perp$. In addition, to prevent a communication
blowup, we also assume the output of $f$ includes the message pattern
$b_k\in\{0, 1\}^{n}$ where the $j$-th bit $b_k[j]$ indicates whether to send message to player
$j$. $\PQ$ will measure the register of $b_k$ and send messages only to
players with $b_k[j]=1$. 

\subparagraph*{Security analysis.} To prove that $\PQ$ is secure against a quantum
full-information adversary, we follow the argument that given any quantum
full-information adversary $\AQ$ attacking $\PQ$, we can construct a
classical adversary $\AC$ in the private-channel model that perfectly simulates $(\PQ,
\AQ)$ when interacting with $\PC$. However, one may question why this
simulation is possible since $\AQ$ is apparently more powerful than $\AC$ in two
aspects: 
\begin{enumerate}
\item $\AQ$ is full-information while $\AC$ is private-channel.
\item $\AQ$ is quantum while $\AC$ is classical.
\end{enumerate}

For the first problem, observe that the randomness of $\PQ$ comes solely from
players' measurement results of $\{b_k\}$ and $\{d_k\}$, of which the
corresponding classical variables in $\PC$ are also available to $\AC$.\footnote{Private-channel $\AC$ knows message patterns by definition. We can
also assume $\AC$ knows the decided values of players because if a uncorrupted player
decides in a BA protocol, all other uncorrupted
players will eventually decide the same value. } The pure state view of $\PQ$
is fully determined by $\{b_k\}$ and $\{d_k\}$, so actually $\AQ$ knows no
more than $\AC$ about the state of the system.

For the second problem, we first consider the Fail-stop adversary case to demonstrate why it is not a concern. The ability
of a Fail-stop adversary $\AQ$ is to halt players and choose a subset of their messages
to be delivered, which is essentially classical. Thus $\AC$ can easily simulate
those actions.

 The Byzantine case is trickier because a Byzantine adversary $\AQ$ can apply quantum
operations on the registers of corrupted players. In this case, we let $\AC$
\emph{classically simulate}\footnote{We assume the adversary is computationally unbounded.}
a quantum state on the registers of corrupted players in order to
keep track of $\AQ$'s actions. 
Moreover, when corrupted players (controlled by $\AC$) send messages to uncorrupted players, they cannot simply transmit quantum messages in the manner $\AQ$ does because  players in $\PC$ are not equipped to receive quantum information. 
To circumvent this challenge, we let corrupted players first measure the messages and then send the measurement outcomes, which are classical, to the uncorrupted players. Intuitively, measuring those messages will not affect the simulation because uncorrupted players always keep a copy of messages they receive. After a quantum message is sent to a uncorrupted player, corrupted players are unable to reobtain it, resulting in the message being traced out from the corrupted players' system, which is equivalent to being measured.
There is still one caveat in the simulation of $\AQ$ by $\AC$: because $\AQ$ is adaptive, it can corrupt new
players during the protocol and reobtain the quantum messages sent to them previously, while $\AC$ will
only obtain collapsed classical messages when corrupting new players. 
To fix this, we let $\AC$ maintain a copy $T$ of the communication transcript between uncorrupted
players and corrupted players. By following this approach, when $\AC$ corrupts some new
players, it can replicate the necessary quantum states as per the content stored in $T$. In this way, $\AC$ can perfectly simulate $\AQ$ in the classical setting.

\subparagraph*{Round and communication complexity.} Our construction of $\AC$
actually yields a stronger result: the probability distribution of executions in $(\PQ, \AQ)$ is identical to that of executions in $(\PC, \AC)$.
This leads to the conclusion presented in \Cref{thm:main}.

\subparagraph*{Extending to the asynchronous model.} Our results in the synchronous
model can be extended to the asynchronous model without extra effort. The primary distinction lies in the measurement metric used; while synchronous protocols are evaluated in \emph{rounds}, asynchronous protocols are evaluated in terms of \emph{steps}. In one step, only one uncorrupted player receives a message, then
performs local computation and possibly sends out messages. It is still feasible to purify
the randomness, perform reversible computation in each step, and develop a quantum
full-information protocol.

Although our results are inspired by the Fail-stop protocol in \cite{10.1145/1060590.1060662}, our techniques are new compared with \cite{10.1145/1060590.1060662}, especially in the Byzantine model. In the Byzantine model, [7] involves an intricate procedure of modifying the original classical protocol by replacing its classical verifiable secret sharing (VSS) component with a quantum VSS. In contrast, our approach focuses on demonstrating the efficacy of extracting purified classical randomness, a feature that is applicable to any classical protocol exhibiting a non-erasing property. Therefore, we expect our technique to have a broader range of applications.


\section{Related Work}

We address the construction of a quantum full-information protocol from a
classical private-channel protocol. In this section, we discuss existing
results in closely related contexts and provide a brief overview of their
techniques.

\subparagraph{BA protocols with private channels.} 
The private-channel model is frequently studied in BA problems. In this model, the adversary is unable to access the contents of the messages
exchanged between the participating players. A seminal work
\cite{doi:10.1137/S0097539790187084} presents a synchronous BA protocol that can
withstand up to $t<n/3$ failures and operates within an expected constant number
of rounds. Additional randomized
protocols \cite{10.1145/800222.806744,10.1145/31846.42229} addressing scenarios
where $n/3 \leq t<n/2$ are known, which require extra assumptions like a
public-key infrastructure and a trusted dealer. Due to their dependency on these
supplementary assumptions, these protocols cannot be adapted to the
information-theoretic setting. In the information-theoretic setting,
\cite{10.5555/983102} presents an asynchronous BA protocol that can withstand up
to $t<n/2$ failures while maintaining a constant running time, particularly
effective against the Fail-stop adversary.  
For the Byzantine adversary, \cite{abraham2008almost} introduces a concept
called shunning verifiable secret sharing and gives an asynchronous BA
protocol with optimal resilience $t<n/3$ and $O(n^2)$ running time, which is
later improved to $O(n)$ by \cite{bangalore2018almost}.

\subparagraph{The full-information model.} The full-information model, as
introduced by \cite{4568166}, serves as a framework for investigating collective
coin-flipping within a network of $n$ players with $t$ failures. This model has
spurred a series of research efforts aimed at enhancing fault tolerance and
reducing round complexity in protocols such as those proposed by \cite{743508}
and \cite{814586}. \cite{doi:10.1137/S0097539793246689} considers the problem of
multiparty computation in the full-information model.
\cite{10.1145/1824777.1824788} gives the first asynchronous leader election
protocol in the full information model with constant success probability against
a constant fraction of corrupted players. Asynchronous BA in the
full-information model used to require exponential time to be solved with linear resilience \cite{ben1983another, bracha1987asynchronous}, which is recently improved to polynomial
time by a sequence of works
\cite{10.1145/2837019,10.1145/3519935.3520015,huang2023byzantine}.

\subparagraph{Quantum Byzantine protocols.} 
Besides the work of \cite{10.1145/1060590.1060662}, many works have applied
quantum principles to Byzantine fault tolerance problems, which has led to
significant advancements in the field. A key contribution is made by
\cite{PhysRevLett.87.217901}, who introduces quantum elements to Byzantine
problems by addressing a weaker version called Detectable Byzantine Agreement
(DBA). Their protocol involves three parties and is based on the Aharonov state.
Building upon this work, \cite{PhysRevLett.100.070504} proposes a 3-party DBA
protocol utilizing four-particle entangled qubits. Further research by
\cite{10.1145/571825.571841} shows that the DBA protocol can reach any tolerance
found. 
Other variants of the problem setting \cite{cholvi2022quantum, luo2019quantum,
PhysRevA.92.042302} are considered to ensure feasibility of the problem against
strong Byzantine adversaries. It is also worth mentioning that a recent work
\cite{hajiaghayi2023brief} improves the communication complexity of the
synchronous Fail-stop protocol of \cite{10.1145/1060590.1060662} from $O(n^2)$
to $O(n^{1+\epsilon})$ for any constant $\epsilon>0$ while maintaining constant
running time.

\section{Preliminaries}

\subsection{Quantum Computation} \label{sec:qc}

In this section, we will briefly discuss quantum computation. For a more
in-depth explanation, readers are encouraged to refer to
\cite{Nielsen_Chuang_2010}. 

In quantum computing, a \emph{qubit} serves as the fundamental unit of quantum
information, analogous to a classical bit. A \emph{pure quantum state} in a
quantum system comprising $n$ qubits, is represented by a unit-length vector in
the $2^n$-dimensional Hilbert space. A commonly used basis of the space is the
\emph{computational basis} $\{\ket{i}=\ket{i_1,i_2,\ldots, i_n}:
i_1,\ldots,i_n\in\{0,1\}\}$. Then any pure state $\ket{\psi}$ can be expressed
as \(\sum _{i=0}^{2^n-1} \alpha_i\ket{i},\) where $\alpha_i$ are complex numbers
known as \emph{amplitudes}, satisfying the condition $\sum _i |\alpha_i|^2=1$.
A \emph{mixed quantum state}, also known as a \emph{density matrix}, represents a
probability mixture of pure states. If a quantum system is in state
$\ket{\psi_i}$ with probability~$p_i$, then its density matrix $\rho:=\sum_i p_i
\ket{\psi_i}\bra{\psi_i}$ where $\bra{\psi_i}$ denotes the conjugate transpose
of $\ket{\psi_i}$. Any density matrix is Hermitian and trace one.
In this paper, we also use density matrix to describe classical probability
distribution: If a random variable $X$ takes value $x_i$ with probability~$p_i$, then it can
be described by the density matrix $\sum_i p_i\ket{x_i}\bra{x_i}$.

Transformations in an $n$-qubit quantum system are described by unitary
transformations in the $2^n$-dimensional Hilbert space. Such a transformation is
depicted by a unitary matrix $U$, which satisfies
$UU^\dagger=\mathbb{I}$ where $^\dagger$ is conjugate transpose and
$\mathbb{I}$ is identity matrix. If $U$ is applied to a pure state $\ket{\psi}$,
the state becomes $U\ket{\psi}$. If $U$ is applied to a mixed state $\rho$, the
state becomes $U\rho U^\dagger$. 

Another important operation is quantum measurement. We will only use projective
measurement in our paper. A projective measurement $\mathcal{M}$ is described by
a collection of orthogonal projectors $\{\Pi_i\}$ such that $\sum_i
\Pi_i=\mathbb{I}$. When
$\mathcal{M}$ is applied on a pure state $\ket{\varphi}$, it collapses to state
$\frac{1}{\sqrt{\beta}}\Pi_i\ket{\varphi}$ with probability
$\beta=\bra{\varphi}\Pi_i\ket{\varphi}$. In the language of density matrix, we
have $\mathcal{M}(\rho)=\sum_i \Pi_i\rho \Pi_i$. In particular, the
\emph{computational basis measurement} has projectors $\{\ket{i}\bra{i}:0\leq
i<2^n\}$. If a quantum state $\sum_{i}\alpha_i\ket{i}$ is measured in computational
basis, it collapses to state $\ket{i}$ with probability~$|\alpha_i|^2$. Measurement can also
be conducted on a portion of the system or on select qubits within the system.
For instance, the measurement restricted on the first qubit of a $n$-qubit system has
projectors $\{\ket{0}\bra{0}\otimes \mathbb{I}_{2^{n-1}}, \ket{1}\bra{1}\otimes
\mathbb{I}_{2^{n-1}}\}$ where $\mathbb{I}_{2^{n-1}}$ is the identity operator on the last $n-1$ qubits.


\subsection{Byzantine Agreement Problem}
In a Byzantine agreement problem, $n$ distinct players labeled from $1$ to $n$ need
to reach a decision on the value of a bit. Each player $i$ inputs a bit $x_i\in
\{0, 1\}$ and must decide an output bit in $\{0, 1\}$ that satisfies the
following conditions:
\begin{enumerate}
    \item {\bf Agreement:} All uncorrupted players decide the same value.
    \item {\bf Validity:} If all $x_i$ are the same bit $y$, then all uncorrupted players decide $y$.
    \item {\bf Termination:} All uncorrupted players terminate with probability $1$.
\end{enumerate}
The problem was introduced by Pease, Shostak and Lamport
\cite{10.1145/322186.322188} in 1980. One can consider different network models,
models of inter-player communication, models of local computation, and fault
models. In this paper, the following models are of interest.
\begin{itemize}
    \item {\bf Network Models.} We will consider both synchronous network, where
    all messages are guaranteed to be delivered within some known time $\Delta$
    from when they are sent, and asynchronous network where messages may be
    arbitrarily delayed.
    \item {\bf Models of Inter-player Communications.} 
    Every two players are connected by a transmit reliable\footnote{Messages will not be corrupted or lost during transmission.} 
    channel. We consider two different communication paradigms: classical and
    quantum. In the classical model, players can communicate classical messages,
    while in the quantum model, they can communicate quantum messages.\footnote{Since classical messages can also be encoded by qubits, no
    additional classical channels are required. }

    \item {\bf Models of Local Computation.} In the field of Byzantine
    protocols, there is a common tendency to overlook the intricacies of local
    computations. We assume players have unbounded computational power and
    local memory.
    \item {\bf Fault models.} We model the faulty behavior of the system by an
    \emph{adversary}. The adversary can corrupt participating players and make them
    deviate from their prescribed programs. Once a player has been corrupted, it remains corrupted permanently. The uncorrupted players are referred to as ``good'' and sometimes the corrupted players are labeled as ``bad''.
     In our
    work, we consider the following types of adversarial behavior:
    \begin{itemize}
        \item {\bf Adaptive.} We will consider \emph{adaptive} adversaries in this paper. An adaptive adversary corrupts players dynamically 
        based on its current information at any time of the protocol.
        \item {\bf Unbound Computation.} Just like good players, the adversary
        has unlimited computational power and memory.
        \item {\bf Private-channel and Full-information.} We will consider
        both private-channel and full-information adversaries. An adversary in the private-channel model is characterized by its lack of adaptation
        based on the specific contents of messages exchanged within a system.
        Essentially, this type of adversary can only discern patterns of
        communication, such as the timing and players involved in message
        exchanges, without access to the actual message contents. By contrast, a
        full-information adversary possesses comprehensive knowledge of all
        local variables associated with the players involved in the system. In
        the context of the quantum model, a full information adversary knows at
        each point the exact pure state of the system.
        \item {\bf Fail-stop and Byzantine.} We will consider both Fail-stop and
        Byzantine adversaries. The players corrupted by the Fail-stop adversary
        will no longer take part in the protocol. We remark that a
        private-channel Fail-stop adversary \emph{cannot} read the local
        memory of corrupted players.\footnote{Some BA protocols consider a
        stronger Fail-stop adversary who can read the memory of corrupted
        players, but our \Cref{thm:main} still applies to those protocols
        because we only require security against a weaker Fail-stop adversary.}
        However, the players corrupted by a Byzantine adversary can deviate arbitrarily
        from the protocol. 
        
    \end{itemize}
\end{itemize}

We are interested in several metrics that measure the performance of BA
protocols:
\begin{itemize}
    \item {\bf Resilience:} the maximum number of parties that can be corrupted
    within the protocol. 
    \item {\bf Round Complexity:} Assume there is a virtual “global clock” within the network that is not accessible to any player. In this context, the term \emph{delay} refers to the time taken from sending a message to its reception. The \emph{number of rounds}\footnote{In the synchronous model, this definition is equivalent to the number of synchronous rounds during the execution.} in an execution refers to the total execution time divided by the longest message delay. The \emph{round complexity} of a protocol $\mathcal{P}$ is defined as the maximum expected number of rounds in $\mathcal{P}$'s executions, considering all inputs and potential adversaries.
    \item {\bf Communication Complexity:} the maximum expected number of messages sent by good players throughout the protocol, considering all inputs and potential adversaries.
\end{itemize}

\subsection{Helper lemmas}

The following two lemmas will be used, of which the proofs are given in
\Cref{sec:proofs}.

\begin{lemma} \label{lem:commute}
    Let $\mathcal{M}$ be the computational basis measurement of a
    Hilbert space $\mathcal{H}$. Then $\mathcal{M}$ commutes with
    \begin{enumerate}
        \item any permutation unitary $U$ acting on $\mathcal{H}$;
        \item any orthogonal projector $\Pi$ on $\mathcal{H}$ in computational
        basis.
    \end{enumerate}
\end{lemma}

\begin{lemma} \label{lem:comm} 

    Let $\mathsf{G}$ be good players' registers, $\mathsf{B}$ be bad players'
    registers. Initially  $\mathsf{G}$ and $\mathsf{B}$ are independent and then they
    make quantum communication for several rounds. Assume $\mathsf{G}$ keeps a local copy of the
    communication transcript between $\mathsf{G}$ and $\mathsf{B}$. Then the pure
    state of the system $\mathsf{GB}$ can be written as
    \(
        \sum_{m} \alpha_{m}\ket{m, \phi_{m}}_\mathsf{G} \otimes \ket{\psi_{m}}_\mathsf{B}
    \)
    where $\ket{m}$ are the communication transcripts, $\ket{\phi_{m}}$ are states
    of $\mathsf{G}$ besides the communication transcripts, and $\ket{\psi_{m}}$ are
    states of $\mathsf{B}$.
\end{lemma}

\section{Proof of Main Theorem} \label{sec:main}

In this section, we prove our main theorem by giving a general reduction from quantum full-information BA protocols to classical private-channel protocols. 

\maintheorem*

Our reduction requires a ``non-erasing'' property of classical private-channel protocols:

\begin{definition}[Non-erasing BA protocol]  \label{def:nonerase}
In the context of a classical BA protocol denoted as $\mathcal{P}$, each computational step performed by a player can be seen as the evaluation of a function $f(s)$ where $s$ is the internal state of the player. Consider a modified protocol, denoted as $\mathcal{P}'$, which follows the structure of $\mathcal{P}$ except that players in $\mathcal{P}'$ keep a copy of their previous state $s$ in their local memory subsequent to each evaluation of $f(s)$. 

A BA protocol such as $\mathcal{P}$ is called \emph{non-erasing} if the adjusted protocol $\mathcal{P}'$ maintains the characteristics of being a BA protocol while preserving the same level of resilience, round and communication complexity as $\mathcal{P}$.
\end{definition}

To the best of our knowledge, this non-erasing property is considered a reasonable assumption as it is met by all existing protocols 
within
the scope of information-theoretic BA with probability one, e.g., 
\cite{chor1989simple,doi:10.1137/S0097539790187084,10.5555/983102,bangalore2018almost}. 
Beyond our scope, there exist BA protocols requiring the ability to securely erase intermediate secrets, often referred to as the \emph{memory-erasure model} \cite{dryja2020lower}. Those protocols either rely on cryptographic assumptions \cite{chen2016algorand} or succeed only with high probability \cite{king2011breaking}.



The rest of this section is to prove \Cref{thm:main}. For simplicity, we will only give a full proof for the synchronous model (\Cref{sec:syn}) and then briefly discuss how to extend it to the asynchronous case (\Cref{sec:async}). 

\subsection{Synchronous Model} \label{sec:syn}


In this subsection, we prove \Cref{thm:main} for the synchronous model. Without loss of generality, we assume a synchronous classical non-erasing private-channel BA protocol $\PC$ has the following normal form.

\subparagraph*{Classical protocol $\PC$.} 
Let $k$ denote the round number, $m_{k}^{(i,j)}$ denote the message sent from $i$ to $j$ and $m_{k}'^{(i,j)}$ denote a copy of $m_{k}^{(i,j)}$ to be kept by $i$, $b_k^{(i, j)}\in\{0,1\}$
denote the message pattern which is $1$ if $m_{k}^{(i,j)}$ is
non-empty, and $d^{(i)}\in\{0,1,\perp\}$ denote the decided value of $i$
($\perp$ if not decided yet). 
We also use $m_{k}^{(\star, i)}$ to denote the
vector $\left(m_{k}^{(1, i)}, m_{k}^{(2, i)}, \ldots, m_{k}^{(n,
i)}\right)$ and $m_{k}'^{(i, \star)}, m_{k}^{(i, \star)}, b_k^{(i,\star)}$ are defined
similarly. At round $k$, player $i$ on input $x_i$ executes the following steps.
\begin{mdframed}
    {\bf $\PC$ for player $i$ at round $k$}
    \begin{enumerate}
        \item Receive messages $m_{k-1}^{(\star, i)}$ from other players if $k>1$.
        \item Sample randomness $r_k^{(i)}$.
        \item Compute a function \(f_P: \Viewrm_{k}^{(i)} \to \left(m_{k}^{(i,
        \star)}, m_{k}'^{(i, \star)}, b_{k}^{(i, \star)}, d_{k}^{(i)}\right)\) where\footnote{Keeping $\Viewrm_k^{(i)}$ in memory does not lose generality beacuse $\PC$ is non-erasing.}
        \begin{align*}
        \Viewrm_{k}^{(i)}&:= \begin{cases}
            \left(i, x_i, r_1^{(i)}\right) & \text{ if } k=1 \\
            \left(\Viewrm_{k-1}^{(i)},  m_{k-1}'^{(i, \star)}, m_{k-1}^{(\star, i)}, r_{k}^{(i)}\right) & \text{ otherwise }
        \end{cases}.
        \end{align*}
        \item If the decided value $d_{k}^{(i)}\not=\perp$, output value
        $d_{k}^{(i)}$ and terminate.\footnote{ We assume a player decides and
        terminates at the same time, since otherwise we can always defer the
        decision until the player terminates. }
        \item \label{item:syncs4} For $j\in[n]$, send messages $m_{k}^{(i,j)}$ to player $j$ if
        $b_k^{(i,j)}=1$.
    \end{enumerate}
\end{mdframed}

Then we construct a quantum BA protocol $\PQ$ by quantizing $\PC$ as
follows. The essential idea is to purify the local randomness, compute
everything reversibly, and do as little measurement as possible. In this way,
only a superposition of all possible local information is revealed to the
quantum full-information adversary. Formally, 

\subparagraph*{Quantum protocol $\PQ$.} 
Let $k$ denote the round number, $\Msf_k^{(i, j)}, \Msf_k'^{(i, j)}, \Bsf_k^{(i, j)}, \Dsf_{k}^{(i)},
\Rsf_{k}^{(i)}$ denote the quantum registers holding the message from player $i$
to player $j$, the copy of the message, the message pattern, the decided value
of player $i$, and the randomness of player $i$ respectively. At round $k$,
player $i$ on input $x_i$ executes the following steps. 
\begin{mdframed}
    {\bf $\PQ$ for player $i$ at round $k$}
    \begin{enumerate}
    \item Receive quantum messages $\Msf_{k-1}^{(\star, i)}$ from other players if $k>1$.
    \item Prepare a quantum state $\sum_{r}\sqrt{\Pr[r_{k}^{(i)}=r]}\ket{r}$ in a
    new quantum register $\mathsf{R}_k^{(i)}$.
    \item Let $U_P^{(i)}$ denote the unitary $\ket{v}\ket{y}\to \ket{v}\ket{y+f_P(v)}$
    which reversibly computes function $f_P$. Execute $U_P$ on register
    $\View_{k}^{(i)}$ and an empty ancilla register
    $\Asf_k^{(i)}:=\left(\Msf_{k}^{(i, \star)}, \Msf_{k}'^{(i, \star)}, \Bsf_{k}^{(i, \star)}, 
    \Dsf_{k}^{(i)}\right)$ where 
    \[
    \View_k^{(i)} :=
    \begin{cases}
        \ket{i}\bra{i}\otimes \ket{x_i}\bra{x_i}\otimes \mathsf{R}_1^{(i)} & \text{ if } k=1 \\
        \left(\View_{k-1}^{(i)}, \Msf_{k-1}'^{(i, \star)}, \Msf_{k-1}^{(\star, i)},  \Rsf_{k}^{(i)}\right) & \text{ otherwise } 
    \end{cases}.
    \]
    \item Measure register $\Dsf_{k}^{(i)}$. If the result
    $d_{k}^{(i)}\not=\perp$, output $d_{k}^{(i)}$ and terminate.
    \item For each $j\in[n]$, measure $\Bsf_{k}^{(i,j)}$. If the result
    $b_k^{(i,j)}=1$, send the $\Msf_{k}^{(i,j)}$ to player $j$. 
    \end{enumerate}
\end{mdframed}

In the rest of this subsection, for both Fail-stop and Byzantine cases, we prove
that $\PQ$ is a quantum full-information BA protocol with the same resilience,
round and communication complexity as $\PC$. 
The proof follows the 
argument that assuming there is quantum full-information adversary $\AQ$
attacking $\PQ$, we can construct a classical adversary $\AC$ in the private-channel model
attacking $\PC$.

\subsubsection{Fail-stop adversary} \label{sec:failsync}
Without loss of generality, we assume the adversary launches attacks at the beginning of
each round for both $\PC$ and $\PQ$. The Fail-stop adversary has the ability to
adaptively halt some players and choose only a subset of their messages in this
round to be received. Now consider a quantum full-information Fail-stop
adversary $\AQ$ attacking $\PQ$, which can be formalized as follows.

\subparagraph*{Quantum full-information adversary $\AQ$.} Assume $\AQ$ samples
its randomness $r_A$ before the protocol starts. Then at round $k$, $\AQ$ first
chooses the set of corrupted players $S_k$ up to round $k$ such that $|S_k|\leq t$
and $S_{k}\supseteq S_{k-1}$, and then $\AQ$ decides only a subset of
$S_k\setminus S_{k-1}$'s messages to be sent. Here, we model the message
exchanging step as a permutation unitary $V_k$ which swaps the registers
$\Msf_k^{(i, j)}$ and the receiving register of player $j$ for $i,j\in[n]$. 
Then $\AQ$'s attack can be modeled by choosing an appropriate $V_k$.
Thus $\AQ$ can be viewed as a function
\(
f_A: r_A, \View_1, \View_2, \ldots,
\View_{k-1} \to (S_k, V_{k})
\)
where $\View_j$ is the pure state view of the system at round $j$.

Let $b_j:=(b_j^{(1,1)}, \ldots, b_j^{(n,n)})$ and $d_j:=(d_j^{(1)}, \ldots,
d_j^{(n)})$. Observe that the randomness of the system comes only from classical
variables $r_A$, $\{b_j\}$, $\{d_j\}$, so the pure state $\View_k$ is fully
determined by those variables. Thus there exists a function $f_V$ such that \(
f_V\left(r_A, b_1, d_1, b_2, d_2, \ldots, b_k, d_k\right)=\View_k. \) Since the
variables $\{b_j\}$ and $\{d_j\}$ in $\PC$ are also available to classical
private-channel adversaries, now we construct a classical private-channel
adversary $\AC$ attacking $\PC$. 

\subparagraph*{Classical adversary $\AC$ in the private-channel model.} 
First sample the same randomness $r_A$ as $\AQ$ before the protocol starts. 
Then at round $k$, compute its action by the following steps.
\begin{enumerate}
    \item For each $j\in[k-1]$, compute quantum state $\ket{\psi_j}:=f_V(r_A,
    b_1, d_1, b_2, d_2, \ldots, b_j, d_j)$.
    \item Compute action $(S_k, V_k):=f_A\left(r_A, \ket{\psi_1}, \ket{\psi_2},
    \ldots, \ket{\psi_{k-1}}\right)$.
\end{enumerate}

Then we prove that $\AC$ perfectly simulates the execution of $(\PQ, \AQ)$ when
interacting with $\PC$, which is characterized by \Cref{lem:syncfail}.

\begin{definition}
    A $k$-round execution $\mathcal{E}$ of $(\PC, \AC)$ is a sequence \(r_A,
    (b_1, d_1), (b_2, d_2), \ldots, (b_k, d_k)\). $\mathcal{E}$ is also a
    $k$-round execution of $(\PQ, \AQ)$ since the pure states of the system at
    each round can completely determined by $\mathcal{E}$ using $f_V$.
\end{definition}

\begin{lemma} \label{lem:syncfail} 
    
    Any $k$-round execution $\mathcal{E}$ occurs in $(\PQ, \AQ)$ and $(\PC,
    \AC)$ with the same probability. Furthermore, if the pure state after
    $\mathcal{E}$ in $(\PQ, \AQ)$ is $\sum_u\alpha_u\ket{u}$, then the
    distribution of the system's possible states after $\mathcal{E}$ in $(\PC,
    \AC)$ conditioned on $\mathcal{E}$ is $\sum_u|\alpha_u|^2\ket{u}\bra{u}$.\footnote{We use density matrix to represent classical probability distribution. See \Cref{sec:qc} for details.}
\end{lemma}

\begin{proof}
    See \Cref{sec:proofsyncfail}.
\end{proof}

By the above lemma, we have:
\begin{proposition} \label{thm:failstop}
    In the synchronous Fail-stop model, given a non-erasing classical private-channel BA protocol $\PC$, there exists a quantum full-information BA protocol $\PQ$ with the same resilience, round and communication complexity as $\PC$.
\end{proposition}
\begin{proof}
    Assuming there exists an adversary $\AQ$ that can cause an inconsistent,
    invalid or non-terminating execution $\mathcal{E}$ in $(\PQ, \AQ)$ with probability
    $p>0$ by corrupting $\leq t$ players, then $\mathcal{E}$ also occurs in $(\PC,
    \AC)$ with probability~$p$ by \Cref{lem:syncfail}, which gives a contradiction. Thus the resilience of $\PQ$ is at least the
    resilience of $\PC$.
    
    Given an execution $\mathcal{E}$, let $|\mathcal{E}|$ be the number of
    rounds of $\mathcal{E}$, and
    $\mathrm{CC}(\mathcal{E}):=\sum_{k=1}^{|\mathcal{E}|}
    \sum_{i\in\bar{S}_k,j\in[n]}b_k^{(i,j)}$ denote the number of messages in
    $\mathcal{E}$. Then by \Cref{lem:syncfail},\footnote{For simplicity, we can assume players' input of is chosen by the adversary, so there is no need to take maximum over the input.} 
    \begin{align*}
    \mathrm{RC}(\PQ)&:=
    \max_{\AQ}
    \E_{\text{execution } \mathcal{E}}
    \Pr[\mathcal{E}\in (\PQ, \AQ)]\cdot |\mathcal{E}|\\
    &\ =\max_{\AC\in \mathcal{Q}}
    \E_{\text{execution } \mathcal{E}}
    \Pr[\mathcal{E}\in (\PC, \AC)]\cdot |\mathcal{E}|
    \leq \mathrm{RC}(\PC), \\
    \mathrm{CC}(\PQ)&:=
    \max_{\AQ}
    \E_{\text{execution } \mathcal{\mathcal{E}}}
    \Pr[\mathcal{E}\in (\PQ, \AQ)]\cdot\mathrm{CC}(\mathcal{E})\\
    &\ =\max_{\AC\in \mathcal{Q}}
    \E_{\text{execution } \mathcal{E}}
    \Pr[\mathcal{E}\in (\PC, \AC)]\cdot \mathrm{CC}(\mathcal{E})
    \leq \mathrm{CC}(\PC)
    \end{align*}
    where $\mathrm{RC}(\cdot)$ denotes round complexity, $\mathrm{CC}(\cdot)$ denotes communication complexity, and $\mathcal{Q}$ denotes the set of classical private-channel
    adversaries that are constructed from some quantum full-information adversary
    in the beyond way.
\end{proof}

\subsubsection{Byzantine adversary} \label{sec:byzsync}

For the Byzantine case, we also assume the adversary launches attacks at the beginning
of each round. Unlike the Fail-stop adversary, the Byzantine adversary can manipulate
corrupted players in an arbitrary way. Now consider a quantum
full-information Byzantine adversary $\AQ$ attacking $\PQ$, which can be
formalized as follows. 

\subparagraph*{Quantum full-information adversary $\AQ$.} Assume $\AQ$ samples
its randomness $r_A$ before the protocol starts. Let $S_k$ denote the corrupted
players up to round $k$ such that $|S_k|\leq t, S_{k}\supseteq S_{k-1}$, and
$\bar{S}_k:=[n]\setminus S_k$ denote good players. Here, we model the message-exchanging step differently from the Fail-stop case. When player $j$ receives the
message from $i$, the register $\Msf_{k}^{(i,j)}$ is simply appended to $j$'s
workspace. Then at round $k$, $\AQ$ acts as follows.
\begin{enumerate}
    \item \label{item:byzqs1} First let current corrupted players $S_{k-1}$ receive all the messages
    sent to them. 
    \item \label{item:byzqs2} Apply arbitrary quantum operation on $S_{k-1}$, which can be
    decomposed as a unitary $U_k$ and a measurement operator $\mathcal{M}_k$ on
    the registers of $S_{k-1}$ by Stinespring dilation theorem.\footnote{
    Stinespring dilation theorem \cite{choi1975completely} states that for any
    quantum operation $\mathcal{E}$, there exists a unitary $U$ and an
    environment space $\mathsf{E}$ such that
    $\mathcal{E}(\rho)=\mathrm{Tr}_{\mathsf{E}}\left(U(\rho\otimes
    \ket{0}\bra{0}_{\mathsf{E}})U^{\dagger}\right)$. The partial trace
    $\mathrm{Tr}_{\mathsf{E}}$ is equivalent to measuring $\mathsf{E}$. We can
    assume players start with large enough empty workspace so there is no need
    to append new ancilla space in order to perform $U$. } Let $a_k$ denote the
    measurement outcome.     
    \item \label{item:byzqs3} Choose an enlarged set $S_k$ of corrupted players and corrupt
    $S_k\setminus S_{k-1}$. 
    \item \label{item:byzqs4} Apply arbitrary quantum operation on $S_{k}$, which can be decomposed
    as applying a unitary $U'_k$ and a measurement operator $\mathcal{M}'_k$ on
    the registers of $S_k$. Let $a'_k$ denote the measurement outcome.
\end{enumerate}

We remark that step~\ref{item:byzqs4} is necessary because 
an adaptive adversary can decide to corrupt a player $i$ and stop (or change) the message just sent by $i$ in step~\ref{item:syncs4} of the previous round.

Similar to the Fail-stop case, the adversary's operations $U_k, \Mcal_k, U'_k,
\Mcal'_k$ and the corrupted set $S_k$ are all functions of randomness $r_A$ and
the system's pure states at each step. And the system's pure states can be fully
determined by classical variables $r_A, \{a_j\}, \{a'_j\}, \{b_j\}$ and
$\{d_j\}$.
Thus we can define two functions $g_A$ and $f_A$ such that
\begin{align*}
    g_A&\left(r_A, a_1, a'_1, b_1, d_1, \ldots, a_{k-1}, a'_{k-1}, b_{k-1}, d_{k-1}\right) = (U_k, \Mcal_k), \text{ and } \\
    f_A&\left(r_A, a_1, a'_1, b_1, d_1, \ldots, a_{k-1}, a'_{k-1}, b_{k-1}, d_{k-1}, a_k\right) = (S_k, U'_k, \Mcal'_k).
\end{align*}

Additionally, we define $\Phi\left(r_A, a_1, a'_1, b_1, d_1, \ldots, a_{k-1},
a'_{k-1}, b_{k-1}, d_{k-1}, a_k\right)$ to be the system's pure state right
after step~\ref{item:byzqs3} of $\AQ$ at round $k$. Then by \Cref{lem:comm}, we have
\begin{equation} \label{eq:indep}
    \Phi\left(r_A, a_1, a'_1, b_1, d_1, \ldots, a_{k-1}, a'_{k-1}, b_{k-1}, d_{k-1}, a_k\right)
    =\sum_{m} \alpha_m\ket{m, \phi_{m}}_{\bar{S}_{k}} \ket{\psi_{m}}_{S_{k}}
\end{equation}
where $\ket{m}$ are the copy of messages between $\bar{S}_{k}$ and $S_{k}$ kept
by $\bar{S}_k$,  $\ket{\phi_{m}}$ are states of $\bar{S}_{k}$ besides the copy, and
$\ket{\psi_{m}}$ are states of $S_{k}$.


Since classical variables $\{b_k\}, \{d_k\}$ in $\PC$ are also available to the adversary in the private-channel model, we can construct a classical Byzantine adversary $\AC$ attacking $\PC$ as follows. 

\subparagraph*{Classical adversary $\AC$ in the private-channel model.} First sample the
same randomness $r_A$ as $\AQ$ before the protocol starts. During the protocol,
$\AC$ maintains a communication transcript $T$ between good players
$\bar{S}_{k}$ and bad players $S_{k}$. Also, $\AC$ \emph{classically simulates}
a quantum state of the registers of $S_{k}$, which is denoted by
$\ket{\varphi_k}$ after round $k$. At round $k$, $\AC$ acts as follows.
\begin{enumerate}
\item \label{item:byzcs1} Let $S_{k-1}$ receive all the messages $m_{k-1}^{(\star, S_{k-1})}$ sent to
them and record in $T$.
\item \label{item:byzcs2} Compute $(U_k, \mathcal{M}_k)$ by $g_A$. Then apply $U_k$ and
$\mathcal{M}_k$ on $\ket{\varphi_{k-1}}\otimes \ket{m_{k-1}^{(\star, S_{k-1})}}$
and obtain the measurement outcome $a_k$.
\item \label{item:byzcs3} Compute $(S_k, U'_k, \mathcal{M}'_k)$ by $f_A$. Corrupt players
$S_k$ and update $T$ as the communication transcript between new sets
$\bar{S}_k$ and $S_k$.
Then according to $T$, $\AC$ discards old state $\ket{\varphi_{k-1}}$ and simulates a
new state $\ket{\psi_{T}}$ which is defined in Eq.~\eqref{eq:indep}.
\item \label{item:byzcs4} Apply $U'_k$ and $\mathcal{M}'_k$ to $\ket{\psi_{T}}$ and obtain
measurement outcome $a'_k$. Then apply a computational basis measurement
$\mathcal{M}_{msg}$ on messages to be sent from $S_k$ to $\bar{S}_k$ and add
those messages to $T$. Let $\ket{\varphi_k}$ be the pure state after applying
$U'_k$, $\mathcal{M}'_k$ and $\mathcal{M}_{msg}$.
\end{enumerate}

\begin{definition}
    A $k$-round execution $\mathcal{E}$ of $(\PC, \AC)$ is a sequence $r_A,
    (a_1, a'_1, b_1, d_1)$, $\ldots$, $(a_k, a'_k, b_k, d_k)$. $\mathcal{E}$ is
    also a $k$-round execution of $(\PQ, \AQ)$ since the pure states of the
    system can be determined by $\mathcal{E}$.
\end{definition}

\begin{lemma} \label{lem:syncbyz}
    Any $k$-round execution $\mathcal{E}$ occurs in $(\PQ, \AQ)$ and $(\PC,
    \AC)$ with the same probability. Furthermore, if the pure state in $(\PQ,
    \AQ)$ after $\mathcal{E}$ is $\ket{Q_k}$, then the distribution of system's
    state in $(\PC, \AC)$ after $\mathcal{E}$ is
    $C_k:=\mathcal{M}_{\bar{S}_k}\left(\ket{Q_k}\bra{Q_k}\right)$
    where
    $\mathcal{M}_{\bar{S}_k}$ is the computational basis measurement on good
    players $\bar{S}_k$'s registers.\footnote{Density matrix $C_k$ represents a distribution of system's states with classical $\bar{S}_k$ and quantum $S_k$. It is classically feasible because $S_k$'s quantum state is classically simulated, and the correlation between $\bar{S}_k$ and $S_k$ is classical, i.e., there is no quantum entanglement.}
\end{lemma}

\begin{proof}
    See \Cref{sec:proofsyncbyz}.
\end{proof}

By the above lemma, we conclude \Cref{thm:main} for the synchronous Byzantine case, which can be proven the same way as the Fail-stop case (\Cref{thm:failstop}).
\begin{proposition}\label{thm:Byzan}
    In the synchronous Byzantine model, given a non-erasing classical private-channel BA protocol $\PC$, there exists a quantum full-information BA protocol $\PQ$ with the same resilience, round and communication complexity as $\PC$.
\end{proposition}

\subsection{Asynchronous Model} \label{sec:async}

The techniques used in the proof above can be extended to the asynchronous model as well. However, a key distinction lies in the terminology used to characterize the execution: while the synchronous model employs ``rounds'', the asynchronous model employs ``steps''. In this context, a step involves a single good player receiving only one message, carrying out computations, and potentially transmitting messages. The order in which players receive messages is determined by the adversary. For simplicity, we assume that each player initially receives its input as its first message, and each message contains the sender's ID.

In alignment with the synchronous model, our approach involves first giving a normal form to any asynchronous classical non-erasing private-channel BA protocol $\PC$ and then quantizing it into a quantum
protocol $\PQ$ against the full-information adversary.


\subparagraph*{Classical protocol $\PC$.} Each player is activated each time it
receives a message. Let $\pi_k^{(i)}$ denote the $k$-th message player $i$
receives, where the first message $\pi_1^{(i)}$ is its input $x_i$. The
notations $m_k^{(i,j)}, m_k'^{(i,j)}, b_k^{(i)}, d_k^{(i)}$ are defined
similarly as in synchronous model (\Cref{sec:syn}).
\begin{mdframed}
    {\bf $\PC$ for player $i$ upon receiving the $k$-th message $\pi_k^{(i)}$}
    \begin{enumerate}
        \item Sample randomness $r_k^{(i)}$.
        \item Compute a function \(f_P: \Viewrm_{k}^{(i)} \to \left(m_{k}^{(i,
        \star)}, m_{k}'^{(i, \star)}, b_{k}^{(i, \star)}, d_{k}^{(i)}\right)\) where  
        \begin{align*}
        \Viewrm_{k}^{(i)}&:= \begin{cases}
            \left(i, x_i, r_1^{(i)}\right) & \text{ if } k=1 \\
            \left(\Viewrm_{k-1}^{(i)}, m_{k-1}'^{(i, \star)}, \pi_k^{(i)}, r_{k}^{(i)}\right) & \text{ otherwise }
        \end{cases}.
        \end{align*}
        \item If the decided value $d_{k}^{(i)}\not=\perp$, output value
        $d_{k}^{(i)}$ and terminate.
        \item For $j\in[n]$, send messages $m_{k}^{(i,j)}$ to player $j$ if
        $b_k^{(i,j)}=1$.
    \end{enumerate}
\end{mdframed}

\subparagraph*{Quantum protocol $\PQ$.} Each player is activated each time it
receives a message. Let $\mathsf{\Pi}_k^{(i)}$ denote the $k$-th quantum message
player $i$ receives. The notations $\mathsf{M}_k^{(i,j)}, \mathsf{M}_k'^{(i,j)},
\Bsf_{k}^{(i, j)}, \mathsf{D}_k^{(i)}$ are defined similarly as in synchronous
model (\Cref{sec:syn}). We remark that $\mathsf{\Pi}_k^{(i)}$ is an alias of
register $\Msf_{k'}^{\left(j', i\right)}$ for some $j', k'$.
\begin{mdframed}
    {\bf $\PQ$ for player $i$ upon receiving the $k$-th message $\mathsf{\Pi}_k^{(i)}$}
    \begin{enumerate}
    \item Prepare a quantum state $\sum_{r}\sqrt{\Pr[r_{k}^{(i)}=r]}\ket{r}$ in a
    new quantum register $\mathsf{R}_k^{(i)}$.
    \item Let $U_P^{(i)}$ denote the unitary $\ket{v}\ket{y}\to \ket{v}\ket{y+f_P(v)}$
    which reversibly computes function $f_P$. Execute $U_P^{(i)}$ on register
    $\View_{k}^{(i)}$ and an empty ancilla register
    $\Asf_k^{(i)}:=\left(\Msf_{k}^{(i, \star)}, \Msf_{k}'^{(i, \star)}, \Bsf_{k}'^{(i, \star)},
    \Dsf_{k}^{(i)}\right)$ where 
    \[
    \View_k^{(i)} :=
    \begin{cases}
        \ket{i}\bra{i}\otimes \ket{x_i}\bra{x_i}\otimes \mathsf{R}_1^{(i)} & \text{ if } k=1 \\
        \left(\View_{k-1}^{(i)}, \Msf_{k-1}'^{(i, \star)}, \mathsf{\Pi}^{(i)}_k,  \Rsf_{k}^{(i)}\right) & \text{ otherwise } 
    \end{cases}.
    \]
    \item Measure register $\Dsf_{k}^{(i)}$. If the result
    $d_{k}^{(i)}\not=\perp$, output $d_{k}^{(i)}$ and terminate.
    \item For each $j\in[n]$, measure $\Bsf_{k}^{(i, j)}$. If the result
    $b_k^{(i,j)}=1$, send the $\Msf_{k}^{(i,j)}$ to player $j$. 
    \end{enumerate}
\end{mdframed}

Then we claim that $\PQ$ is a quantum BA protocol against the quantum
full-information adversary in the asynchronous model with the same round and
communication complexity as $\PC$. The proof is almost the same as the
synchronous case, so we only sketch the proof here. 

Assuming there is a quantum full-information Fail-stop (resp.~Byzantine) adversary
$\AQ$ attacking $\PQ$, we can construct a classical private-channel Fail-stop
(resp.~Byzantine) adversary $\AC$ attacking $\PC$ as in \Cref{sec:failsync}
(resp.~\Cref{sec:byzsync}). Then we can define execution execution in the asynchronous model.
\begin{definition}[Informal]
    A $k$-step execution $\mathcal{E}$ is defined to be a sequence 
    $r_A,$ $(a_1, b_1, d_1),$ $(a_2, b_2, d_2), \ldots, (a_k, b_k, d_k)$
    where $r_A$ is the adversary's randomness, $a_j$ is some classical
    information the adversary obtains at step $j$, $b_j\in \{0,1\}^n$ is the message
    pattern, and $d_j\in\{0, 1,\perp\}$ is the decided value of the player
    activated at step $j$.
\end{definition}

Then similar to \Cref{lem:syncfail} (resp.~\Cref{lem:syncbyz}), we prove that any
execution occurs in $(\PQ, \AQ)$ with the same probability.
\begin{lemma}[Informal] \label{lem:async}
    Any $k$-step execution $\mathcal{E}$ occurs in $(\PQ, \AQ)$ and $(\PC, \AC)$
    with the same probability. 
\end{lemma}

Since the information contained in an execution $\mathcal{E}$ fully determines
the properties, number of rounds, and number of messages of the protocol, we can conclude \Cref{thm:main} in the asynchronous case. This can be proven similarly to the synchronous case (\Cref{thm:failstop} and \Cref{thm:Byzan}).
\begin{proposition}\label{thm;async}
    In the asynchronous Fail-stop (or Byzantine) model, given a non-erasing classical private-channel BA protocol $\PC$, there exists a quantum full-information BA protocol $\PQ$ with the same resilience, round and communication complexity as $\PC$.
\end{proposition}

\begin{proof}[Proof of \Cref{thm:main}]
    \Cref{thm:main} can be obtained by integrating the results from \Cref{thm:failstop}, \Cref{thm:Byzan} and \Cref{thm;async}.
\end{proof}




\section{Discussions}
In this paper, we present a general reduction from quantum full-information BA
protocols to classical private-channel BA protocols that preserves resilience,
round and communication complexity. Utilizing this reduction, we make progress
towards the open question posed by \cite{10.1145/1060590.1060662} of whether
quantum BA can achieve $O(1)$ round complexity and optimal resilience
$t<n/3$ simultaneously in the asynchronous full-information model. We show that
$O(1)$ round complexity and suboptimal resilience $t<n/(3+\epsilon)$ is possible
for any constant $\epsilon>0$. Our reduction also suggests that designing
a better classical private-channel protocol may finally lead to the resolution
of this open question.

There are several interesting directions for future research. Firstly, it would be valuable to explore whether the reverse of our reduction is possible, i.e., whether any
quantum full-information BA protocol can be converted to a classical
private-channel BA protocol without compromising key attributes like resilience. Existing techniques in this paper do not apply
due to the ability of good players to employ quantum operations. Secondly, it is worth considering the potential generalization of our results to less strict models, such as BA that terminates only with high
probability \cite{canetti1993fast, king2011breaking}, or BA that requires
erasing intermediate states \cite{king2011breaking, chen2016algorand}. 
Thirdly, it is worthwhile to explore the potential for developing BA protocols with improved performance by granting quantum players the ability to utilize private memory, thereby shifting the adversary from a position of full-information to one of limited knowledge. This model presents an intriguing opportunity for innovation, especially considering the existence of quantum key distribution in such a framework \cite{bennett2014quantum}.
Finally, while our primary focus is on addressing the BA problem as it stands as a fundamental challenge in this field, we anticipate that our methods can also be applied to other fault-tolerant distributed computing tasks like coin toss and leader election.

\bibliography{citations}

\appendix

\section{Proofs of helper lemmas}
\label{sec:proofs}

\subsection{Proof of \Cref{lem:commute}}

\begin{proof}
    \begin{enumerate}
    \item Since $U$ is a permutation unitary, there exists a permutation $\pi$
    such that $U\ket{i}=\ket{\pi(i)}$ for any computational basis $\ket{i}$ in
    $\mathcal{H}$. Given any density matrix $\rho=\sum_{i,j}
    \rho_{i,j}\ket{i}\bra{j}$ in $\mathcal{H}$, we have
    \begin{align*}
        \mathcal{M}(U\rho U^\dagger)
        &=\mathcal{M}\left(\sum_{i,j}\rho_{i,j}\ket{\pi(i)}\bra{\pi(j)}\right)\\
        &=\sum_{k} \ket{k}\bra{k} \sum_{i,j}\rho_{i,j}\ket{\pi(i)}\bra{\pi(j)} \ket{k}\bra{k}=
        \sum_{k}\rho_{k,k}\ket{\pi(k)}\bra{\pi(k)},\\
        U\mathcal{M}(\rho)U^\dagger 
        &= U\sum_{k}\ket{k}\bra{k} \sum_{i,j}\rho_{i,j}\ket{i}\bra{j} \ket{k}\bra{k}U^\dagger \\
        &= U\sum_{k}\rho_{k,k}\ket{k}\bra{k}U^\dagger
        = \sum_{k}\rho_{k,k}\ket{\pi(k)}\bra{\pi(k)}.
    \end{align*}
    Thus $\mathcal{M}(U\rho U^\dagger)=U\mathcal{M}(\rho)U^\dagger$.
    \item Since $\Pi$ is an orthogonal projector in the computational basis, we have
    $\Pi=\sum_{i\in S} \ket{i}\bra{i}$ for some set $S$. For any state
    $\rho\in\mathcal{H}$, one can verify that
    $\mathcal{M}(\Pi\rho\Pi^\dagger)=\Pi\mathcal{M}(\rho)\Pi^\dagger$.
    \end{enumerate}
\end{proof}

\subsection{Proof of \Cref{lem:comm}}

\begin{proof}
    Proof by induction on the number of messages. Initially, $\mathsf{G}$ and $\mathsf{B}$ are
    independent, so the state of $\mathsf{GB}$ is $\ket{\phi_0}_\mathsf{G}\otimes
    \ket{\psi_0}_\mathsf{B}$. Assuming currently the state of $\mathsf{GB}$ is
    $\sum_{m} \alpha_{m}\ket{m, \phi_{m}}_\mathsf{G} \otimes \ket{\psi_{m}}_{\sf
    B}$, consider the next message. First $\mathsf{G}$ and $\mathsf{B}$ apply a local
    unitary $U_\mathsf{G}\otimes U_\mathsf{B}$ to generate messages. Note that $U_{\sf
    G}$ will not change the previous transcripts $\ket{m}$. Then the state
    becomes
    \[
    \sum_{m} \alpha_{m}U_\mathsf{G}\ket{m, \phi_{m}}_\mathsf{G} \otimes U_\mathsf{B}\ket{\psi_{m}}_\mathsf{B}=
    \sum_{m} \alpha_{m}\ket{m, \phi_{m}'}_\mathsf{G} \otimes \ket{\psi_{m}'}_\mathsf{B}.
    \]

    \begin{itemize}
    \item If the message is sent by $\mathsf{G}$, then the system can be written as
    \[
    \sum_{m} \alpha_{m}\sum_{m'}\beta_{m'}\ket{m, m', m', \phi_{m,m'}'}_\mathsf{G} \otimes \ket{\psi_{m}'}_\mathsf{B}
    \]
    where the second $m'$ is the message to be sent to $\mathsf{B}$ and the first
    $m'$ is a copy to be kept by $\mathsf{G}$. After sending the message, the system becomes
    \[
    \sum_{m,m'} \alpha_{m}\beta_{m'}\ket{m, m', \phi_{m,m'}'}_\mathsf{G} \otimes \ket{m', \psi_{m}'}_\mathsf{B}.
    \]
    \item If the message is sent by $\mathsf{B}$, then the system can be written as
    \[
    \sum_{m} \alpha_{m}\ket{m, \phi_{m}}_\mathsf{G} \otimes \sum_{m'}\beta_{m'}\ket{m', \psi_{m,m'}'}_\mathsf{B}.
    \]
    After sending the message, the system becomes
    \[
    \sum_{m, m'} \alpha_{m}\beta_{m'} \ket{m, m', \phi_{m}}_\mathsf{G} \otimes\ket{\psi_{m,m'}'}_\mathsf{B}.    
    \]
    \end{itemize}
\end{proof}

\section{Proof of \Cref{lem:syncfail}}
\label{sec:proofsyncfail}

\begin{proof}
    Prove by induction on $k$. When $k=0$, the execution $\mathcal{E}_0=r_A$
    occurs in $(\PQ, \AQ)$ and $(\PC, \AC)$ both with the probability of $r_A$.
    Assume the lemma holds for $k-1$. 
    Consider a $k$-round execution \(\mathcal{E}_k := r_A, (b_1, d_1), (b_2,
    d_2), \ldots, (b_{k-1}, d_{k-1}), (b_k, d_k).\) By inductive hypothesis, the
    $(k-1)$-round prefix \(\mathcal{E}_{k-1} := r_A, (b_1, d_1), (b_2, d_2),
    \ldots, (b_{k-1}, d_{k-1})\) occurs with probability~$p$ in both $(\PQ, \AQ)$ and $(\PC,
    \AC)$, and the state before round $k$ is
    $\ket{Q_{k-1}}:=\sum_u\alpha_u\ket{u}$ and
    $C_{k-1}:=\sum_u|\alpha_u|^2\ket{u}\bra{u}$ for $(\PQ, \AQ)$ and $(\PC,
    \AC)$ respectively. 

    \proofsubparagraph*{Round $k$ of $(\PQ, \AQ)$.}
    After $\AQ$'s action and players receiving messages, the state of the system
    becomes \(V_k\ket{Q_{k-1}}\). Then good players first prepare a
    superposition state $\ket{r_k}$ of randomness in a new register $\Rsf_k$ and
    prepare $\ket{0}$ in a new register $\Asf_k$. Note that here
    $\Rsf_k:=(\Rsf_k^{(1)},\ldots, \Rsf_k^{(n)})$,
    $\Asf_k:=(\Asf_k^{(1)},\ldots, \Asf_k^{(n)})$ and all other notations
    without superscript are defined similarly.
    
    Then the players apply the
    unitary operator $U_P:=\otimes_{i=1}^n U_P^{(i)}$ followed by a measurement which outputs $(b_k, d_k)$.
    The measurement can be viewed as an orthogonal projector $\Pi_{b_k, d_k}$ in
    computational basis that projects the quantum state of registers $(\Bsf_k,
    \Dsf_k)$ into values $(b_k, d_k)$. Then the state after round $k$ becomes
    \[
        \ket{Q_k}:=\frac{1}{\sqrt{\beta}}\Pi_{b_k, d_k}U_P\left(V_k\ket{Q_{k-1}}\otimes \ket{0}_\Asf\otimes \ket{r_k}_\Rsf\right)
    \] 
    where $\beta$ is the probability of getting measurement outcome $(b_k, d_k)$.

    
    \proofsubparagraph*{Round $k$ of $(\PC, \AC)$.} The first
    observation is that $C_k$ can be viewed as first applying the same operation
    as $(\PQ, \AQ)$ and then applying computational basis measurement
    $\mathcal{M}$ on the whole system:
    \[
    C_k:= \mathcal{M}
    \left(
    \frac{1}{\beta' }  \Pi_{b_k, d_k} U_P V_k
    \left(C_{k-1}\otimes \ket{0,r_k}\bra{0,r_k}_{\Asf\Rsf}
    \right)
    V_k^{\dagger} U_P^\dagger \Pi_{b_k, d_k}^\dagger
    \right)
    \]
    where $\beta'$ is the probability of getting measurement outcome $(b_k,
    d_k)$. The second observation is that $C_{k-1}=\mathcal{M}'(\ket{Q_{k-1}}\bra{Q_{k-1}})$
    where $\mathcal{M}'$ denotes the computational basis measurement in $\ket{Q_{k-1}}$'s space. Then
    \begin{align*}
        C_k
        &=\mathcal{M}\left(
        \frac{1}{\beta' }  \Pi_{b_k, d_k} U_P V_k
        \left(\mathcal{M}'\left(\ket{Q_{k-1}}\bra{Q_{k-1}}\right)\otimes\ket{0,r_k}\bra{0,r_k}_{\Asf\Rsf} \right)
        V_k^{\dagger} U_P^\dagger \Pi_{b_k, d_k}^\dagger
        \right) \\
        &=
        \frac{1}{\beta' }  \Pi_{b_k, d_k} U_P V_k
        \mathcal{M}\left(\mathcal{M}'\left(\ket{Q_{k-1}}\bra{Q_{k-1}}\right)\otimes\ket{0,r_k}\bra{0,r_k}_{\Asf\Rsf} \right)
        V_k^{\dagger} U_P^\dagger \Pi_{b_k, d_k}^\dagger.
    \end{align*}
    The second equality is because $U_P, V_k$ are all permutation unitaries and
    $\Pi_{b_k, d_k}$ is an orthogonal projector in computational basis, which
    all commute with $\mathcal{M}$ by \Cref{lem:commute}. Since $\mathcal{M}$
    measures a larger space than $\mathcal{M}'$, $\mathcal{M}'$ can be absorbed
    into $\mathcal{M}$, i.e., $\mathcal{M}\mathcal{M}'\equiv \mathcal{M}$. Thus
    \begin{align*}
        C_k
        &=
        \frac{1}{\beta' }  \Pi_{b_k, d_k} U_P V_k
        \mathcal{M}\left(\ket{Q_{k-1}}\bra{Q_{k-1}}\otimes\ket{0, r_k}\bra{0, r_k}_{\Asf\Rsf} \right)
        V_k^{\dagger} U_P^\dagger \Pi_{b_k, d_k}^\dagger \\
        &=
        \mathcal{M}\left(\frac{1}{\beta' }  \Pi_{b_k, d_k} U_P V_k
        \left(\ket{Q_{k-1}}\bra{Q_{k-1}}\otimes\ket{0, r_k}\bra{0, r_k}_{\Asf\Rsf} \right)
        V_k^{\dagger} U_P^\dagger \Pi_{b_k, d_k}^\dagger \right)\\
        &=\frac{\beta}{\beta'}\mathcal{M} \left(\ket{Q_k}\bra{Q_k}\right).
    \end{align*}
    Finally, we have $\beta=\beta'$ because $C_k$ has trace $1$. Thus
    $C_k=\mathcal{M}\left(\ket{Q_k}\bra{Q_k}\right)$ and the probability of the
    $\mathcal{E}_k$ occurring is $p\beta$ for both $(\PQ, \AQ)$ and $(\PC,
    \AC)$.
\end{proof}

\section{Proof of \Cref{lem:syncbyz}} \label{sec:proofsyncbyz}
\begin{proof}
    Prove by induction on $k$. The base case $k=0$ is trivial. Assume the
    proposition holds for $k-1$. Consider a $k$-round execution \(\mathcal{E}_k
    := r_A, (a_1, a'_1, b_1, d_1), \ldots, (a_k, a'_k, b_k, d_k).\) By inductive
    hypothesis, the $(k-1)$-round prefix \(\mathcal{E}_{k-1} := r_A, (a_1, a'_1,
    b_1, d_1), \ldots, (a_{k-1}, a'_{k-1}, b_{k-1}, d_{k-1})\) occurs with probability~$p$
    in both $(\PQ, \AQ)$ and $(\PC, \AC)$, and the state before round $k$ is
    $\ket{Q_{k-1}}$ and
    $C_{k-1}:=\mathcal{M}_{S_{k-1}}\left(\ket{Q_{k-1}}\bra{Q_{k-1}}\right)$ for
    $(\PQ, \AQ)$ and $(\PC, \AC)$ respectively. Since good players' messages are
    fully determined by their local variables, $\mathcal{M}_{S_{k-1}}$ can be
    restricted to a measurement $\mathcal{M}_{S_{k-1}}'$ which does not measure
    the messages $\bar{S}_{k-1}$ are about to send out, i.e.,
    $\Msf_k^{(\bar{S}_{k-1}, \star)}$. Then
    $C_{k-1}:=\mathcal{M}_{S_{k-1}}'\left(\ket{Q_{k-1}}\bra{Q_{k-1}}\right)$. In
    the following, we consider the evolution of $\ket{Q_{k-1}}$ and $C_{k-1}$ in round $k$.

    \proofsubparagraph*{Step~\ref{item:byzqs1} and \ref{item:byzqs2} of the adversary.} Both $\AQ$ and $\AC$
    apply $U_k$ and $\mathcal{M}_k$ on $S_{k-1}$ along with the messages
    $\Msf_{k-1}^{(\star, S_{k-1})}$ sent to them. Since we know $\mathcal{M}_k$ will
    output $a_k$, it can be viewed as an orthogonal projector $\Pi_{a_k}$ that
    projects into the space of $a_k$. Since $\Pi_{a_k}$ and $U_k$ act only on
    $S_{k-1}$'s registers and the messages $\bar{S}_{k-1}$ will send to
    $S_{k-1}$, they commute with $\mathcal{M}'_{\bar{S}_{k-1}}$
    . Thus the states of $(\PQ, \AQ)$ and $(\PC, \AC)$ become
    \begin{align*}
    \ket{Q_{k-0.5}} &:= \frac{1}{\sqrt{\gamma}}\Pi_{a_k}U_k\ket{Q_{k-1}} \text{, and } \\
    C_{k-0.5} &:= \frac{1}{\gamma'}\Pi_{a_k}U_k C_{k-1} U_k^\dagger \Pi_{a_k}^\dagger
    = \frac{1}{\gamma'}\Pi_{a_k}U_k \mathcal{M}_{\bar{S}_{k-1}}'\left(\ket{Q_{k-1}}\bra{Q_{k-1}}\right) U_k^\dagger \Pi_{a_k}^\dagger \\
    &\ = \mathcal{M}_{\bar{S}_{k-1}}'\left(\frac{1}{\gamma'}\Pi_{a_k}U_k \ket{Q_{k-1}}\bra{Q_{k-1}} U_k^\dagger \Pi_{a_k}^\dagger\right)
    = \frac{\gamma}{\gamma'}\mathcal{M}_{\bar{S}_{k-1}}'(\ket{Q_{k-0.5}}\bra{Q_{k-0.5}}).
    \end{align*}
    where $\gamma$ and $\gamma'$ are probabilites of $\ket{Q_{k-1}}$ and $C_{k-1}$
    outputting $a_k$. Since $C_{k-0.5}$ has trace $1$, we have $\gamma=\gamma'$
    and thus $C_{k-0.5}=\mathcal{M}'_{\bar{S}_{k-1}}\left(\ket{Q_{k-0.5}}\bra{Q_{k-0.5}}\right)$.

    \proofsubparagraph*{Step~\ref{item:byzqs3} of the adversary.} Both $\AQ$ and $\AC$ choose an
    enlarged set $S_k$ of corrupted players. This step does not affect the state
    $\ket{Q_{k-0.5}}$ of $(\PQ, \AQ)$. By \Cref{lem:comm}, we have
    $\ket{Q_{k-0.5}}=\sum_{m}\alpha_m\ket{m,
    \phi_{m}}_{\bar{S}_{k}}\ket{\psi_{m}}_{S_{k}}$. Since
    $\Mcal'_{\bar{S}_{k-1}}$ can be decomposed as $\Mcal'_{\bar{S}_{k-1}}\equiv
    \Mcal_1\otimes \Mcal_2 \otimes \Mcal_3$, where $\Mcal_1$ acts on transcript
    $\ket{m}$, $\Mcal_2$ acts on registers of $\bar{S}_{k-1}$ besides $\ket{m}$,
    and $\Mcal_3$ acts on registers newly corrupted players $S_k\setminus
    S_{k-1}$, we have
    \[
    C_{k-0.5}=\mathcal{M}_{\bar{S}_{k-1}}'\left(\ket{Q_{k-0.5}}\bra{Q_{k-0.5}}\right)
    =\sum_m |\alpha_m|^2  \ket{m}\bra{m} \otimes \mathcal{M}_2\left(\ket{\phi_{m}}\bra{\phi_{m}}\right)
    \otimes \mathcal{M}_3\left(\ket{\psi_m}\bra{\psi_m}\right).
    \]
    In step~\ref{item:byzcs3} of $\AC$, $\AC$ has recorded the transcript $m$ and will
    discard the old state $\mathcal{M}_3\left(\ket{\psi_m}\bra{\psi_m}\right)$ and simulate a
    new state $\ket{\psi_{m}}$. After that, the state of $(\PC, \AC)$ becomes
    \[
    C'_{k-0.5}:=
    \sum_m |\alpha_m|^2  \ket{m}\bra{m} \otimes \mathcal{M}_2\left(\ket{\phi_{m}}\bra{\phi_{m}}\right)
    \otimes\ket{\psi_m}\bra{\psi_m}
    = \mathcal{M}_{\bar{S}_{k}}' \left(\ket{Q_{k-0.5}}\bra{Q_{k-0.5}}\right).
    \]
    Note that we use $\mathcal{M}_{\bar{S}_{k}}'$ to distinguish from operator
    $\mathcal{M}_{\bar{S}_{k}}$ which also measures the newly appended registers $\Asf_k,
    \Rsf_k$, and $\Msf_{k-1}^{(S_k, \bar{S}_k)}$ of $\bar{S}_k$ at round
    $k$.

    \proofsubparagraph*{Step~\ref{item:byzqs4} of the adversary and good players' action.} Step~\ref{item:byzqs4} of $\AQ$ applies $U'_k$ followed by measurement $\Mcal'_k$ which outputs
    $a'_k$ on $S_k$'s registers. The measurement can be viewed as a projector
    $\Pi_{a'_k}$ that projects into the space of $a'_k$. Then good players apply
    unitary $U_P$ and projector $\Pi_{b_k, d_k}$ which projects the state of
    registers $(\Bsf_k, \Dsf_k)$ into values $(b_k, d_k)$. Thus the state of
    $(\PQ, \AQ)$ after round $k$ becomes
    \[
        \ket{Q_k}:=\frac{1}{\sqrt{\beta}}\Pi_{b_k, d_k}U_P\left(\Pi_{a_k'}U'_k\ket{Q_{k-0.5}}\otimes \ket{0}_\Asf\otimes \ket{r_k}_\Rsf\right).
    \]
    where $\beta$ is the probability of outputting $a'_k, b_k$, and $d_k$.

    Step~\ref{item:byzcs4} of $\AC$ will additionally apply a measurement $\mathcal{M}_{msg}$ on
    messages $\Msf_{k-1}^{(S_k, \bar{S}_k)}$ sent from $S_k$ to $\bar{S}_k$. The good players' action of $\PC$
    can be viewed as applying the same operation as $\PQ$ and then measuring good
    players $\bar{S}_k$'s space in computational basis. Thus the state of $(\PC,
    \AC)$ becomes
    \begin{align*}
        C_k &:=\mathcal{M}_{\bar{S}_k}\left(\frac{1}{\beta'}\Pi_{b_k, d_k}U_P\left(\mathcal{M}_{msg}(\Pi_{a_k'}U'_kC'_{k-0.5}U_k'^\dagger\Pi_{a_k'}^\dagger)\otimes \ket{0, r_k}\bra{0, r_k}_{\Asf\Rsf}\right)U_P^\dagger\Pi_{b_k, d_k}^\dagger\right)\\
        &\ =\mathcal{M}_{\bar{S}_k}\left(\frac{1}{\beta'}\Pi_{b_k, d_k}U_P\left(\mathcal{M}_{msg}(\Pi_{a_k'}U'_k\mathcal{M}_{\bar{S}_{k}}' \left(\ket{Q_{k-0.5}}\bra{Q_{k-0.5}}\right)U_k'^\dagger\Pi_{a_k'}^\dagger)\otimes \ket{0, r_k}\bra{0, r_k}_{\Asf\Rsf}\right)U_P^\dagger\Pi_{b_k, d_k}^\dagger\right)
    \end{align*}
    where $\beta'$ is the probability of outputting $a'_k, b_k$, and $d_k$. Similar
    to Fail-stop case, $\mathcal{M}_{\bar{S}_k}$ and $U_P\Pi_{b_k, d_k}$ commute
    by \Cref{lem:commute}. $\mathcal{M}'_{\bar{S}_k}$ and
    $\Pi_{a_k'}U'_k$ also commute because they act on $\bar{S}_k$ and $S_k$
    separately. Thus
    \begin{align*}
        C_k&=\frac{1}{\beta'}\Pi_{b_k, d_k}U_P\mathcal{M}_{\bar{S}_k}\mathcal{M}_{msg}\mathcal{M}_{\bar{S}_{k}}'\left(\Pi_{a_k'}U'_k \ket{Q_{k-0.5}}\bra{Q_{k-0.5}}U_k'^\dagger\Pi_{a_k'}^\dagger\otimes \ket{0, r_k}\bra{0, r_k}_{\Asf\Rsf}\right)U_P^\dagger\Pi_{b_k, d_k}^\dagger \\
        &=\frac{1}{\beta'}\Pi_{b_k, d_k}U_P\left(\mathcal{M}_{\bar{S}_k}(\Pi_{a_k'}U'_k \ket{Q_{k-0.5}}\bra{Q_{k-0.5}}U_k'^\dagger\Pi_{a_k'}^\dagger)\otimes \ket{0, r_k}\bra{0, r_k}_{\Asf\Rsf}\right)U_P^\dagger\Pi_{b_k, d_k}^\dagger \\
        &=\mathcal{M}_{\bar{S}_k}\left(\frac{1}{\beta'}\Pi_{b_k, d_k}U_P\left(\Pi_{a_k'}U'_k \ket{Q_{k-0.5}}\bra{Q_{k-0.5}}U_k'^\dagger\Pi_{a_k'}^\dagger\otimes \ket{0, r_k}\bra{0, r_k}_{\Asf\Rsf}\right)U_P^\dagger\Pi_{b_k, d_k}^\dagger\right) \\
        &=\frac{\beta}{\beta'}\mathcal{M}_{\bar{S}_k} \left(\ket{Q_k}\bra{Q_k}\right).
    \end{align*}
    where the second equality is because $\mathcal{M}_{\bar{S}_k}$ measures a
    larger space than $\mathcal{M}_{msg}\mathcal{M}'_{\bar{S}_k}$, thus
    $\mathcal{M}_{\bar{S}_k}\mathcal{M}_{msg}\mathcal{M}'_{\bar{S}_k}\equiv \mathcal{M}_{\bar{S}_k}$.

    Finally, since $C_k$ has trace $1$, we have $\beta=\beta'$ and thus
    $C_k=\mathcal{M}_{\bar{S}_k}\left(\ket{Q_k}\bra{Q_k}\right)$. The probability of
    $\mathcal{E}_k$ occurring is $p\gamma\beta$ for both cases.
\end{proof}

\section{Bit-efficient Asynchronous Protocols}
\label{sec:bit-efficient}
We give two explicit constructions of asynchronous quantum BA protocol against the
full-information adversary. Both of them are quantized from some classical
private-channel protocol. Stronger than the general reduction in the main
text, these two constructions not only preserve the round and communication
complexity, but also is bit-efficient. In other words, the local computational
complexity and number of communication bits is also the same as the original classical
protocol.

\subsection{Fail-stop case}
We quantize the classical protocol in Section 14.3 of
\cite{10.5555/983102}. The original classical protocol is a phase-based
voting scheme: Each phase consists of two components, voting for their
preferences and tossing a common coin. Since voting is fully deterministic, we
only need to quantize the common-coin protocol. In the following, we present a
quantum common-coin protocol (\Cref{algo:coin}) and then prove that it achieves
properties of common coin against quantum full-information adversary that can
corrupt up to $t<n/2$ players (\Cref{lem:coin}).

\begin{algorithm}[!htb]
    \caption{Quantum common-coin protocol: code for player $i$.} \label{algo:coin}
    \begin{algorithmic}[1]
        \Procedure{\textsc{Quantum-common-coin}}{}
        \State $\ket{c}\gets \sqrt{\frac{1}{n}}\ket{0}+\sqrt{1-\frac{1}{n}}\ket{1}$ \label{line:coin}
        \State $\mathsf{coins}\gets \textsc{Quantum-get-core}\left(\ket{c}\right)$
        \State apply gate $C^n X$ on $\mathsf{coins}\otimes \ket{0}\bra{0}_\Asf$ where
        \[
        C^n X\ket{x_1, x_2, \ldots, x_n, y}:= \begin{cases}
            \ket{x_1, x_2, \ldots, x_n, y\oplus 1} & \text{ if } x_1=x_2=\cdots=x_n=1 \\
            \ket{x_1, x_2, \ldots, x_n, y} & \text{ else }
        \end{cases}          
        \]
        \State \label{line:measure} measure register $\Asf$ and output the measurement result as coin value
        \EndProcedure
        \item[]
        \Procedure{\textsc{Quantum-get-core}}{$\mathsf{coin}$}
            \State $S_1,S_2,S_3\gets \emptyset$; $\mathsf{coins}[j]\gets \perp$ for $j\in[n]$
            \State $\mathsf{coins}[i]\gets \mathsf{coin}$
            \State $\textsc{Quantum-multicast}(\langle \textrm{first}, \mathsf{coins}[i]\rangle)$
            \Upon{$\langle \textrm{first}, \mathsf{v}\rangle$ is received from player $j$}
            \State $\mathsf{coins}[j]\gets \mathsf{v}$
            \State $S_1\gets S_1\cup \{j\}$
            \If {$|S_1|=n-t$}
                {$\textsc{Quantum-multicast}(\langle \textrm{second}, \mathsf{coins}\rangle)$}
            \EndIf
            \EndUpon
            \Upon{$\langle \textrm{second}, \mathsf{values}\rangle$ is received from player $j$}
            \State $S_2\gets S_2\cup \{j\}$
            \State $\mathsf{coins}[j]\gets \mathsf{values}[j]$ if $\mathsf{coins}[j]=\perp$ for $j\in[n]$
            \If {$|S_2|=n-t$}
            {$\textsc{Quantum-multicast}(\langle \textrm{third}, \mathsf{coins}\rangle)$}
            \EndIf
            \EndUpon
            \Upon{$\langle \textrm{third}, \mathsf{values}'\rangle$ is received from player $j$}
            \State $S_3\gets S_3\cup \{j\}$
            \State $\mathsf{coins}[j]\gets \mathsf{values}'[j]$ if $\mathsf{coins}[j]=\perp$ for $j\in[n]$
            \If {$|S_3|=n-t$}
            \State $\mathsf{coins}[j]\gets \ket{1}$ if $\mathsf{coins}[j]=\perp$ for $1\leq j\leq n$
            \State {\bf return} $\mathsf{coins}$
            \EndIf
            \EndUpon
        \EndProcedure
        \item[]
        \Procedure{\textsc{Quantum-multicast}}{$\langle\textrm{message}, \mathsf{qubits}\rangle$}
            \For{$j \in [n]$}
            \State initialize an empty quantum register $\mathsf{qubits}'$
            \State apply $CX$ gate on $\mathsf{qubit}\otimes \mathsf{qubits}'$ where
            \(
             CX \ket{a}\ket{b} = \ket{a}\ket{a+b}
            \)
            \State send $\langle\textrm{message}, \mathsf{qubits}'\rangle$ to player $j$
            \EndFor
        \EndProcedure
    \end{algorithmic}    
\end{algorithm}

\begin{lemma} \label{lem:core}
    In \textsc{Quantum-get-core}, there exists a set of players $C$,
such that $|C|\geq n-t$ and the coin of every player in $C$ is received by all
good players.
\end{lemma}

\begin{proof}
    The proof is the same as Lemma 14.5 in \cite{10.5555/983102} 
    for the classical \textsc{Get-core} procedure.
\end{proof}

\begin{lemma} \label{lem:coin}
    Against a quantum full-information Fail-stop adversary that can corrupt up
    to $t<n/2$ players, \textsc{Quantum-common-coin} satisfies the following properties.
    \begin{enumerate}
    \item All good players terminate in $O(1)$ rounds.
    \item The probability of all good players outputting $1$ is $\Omega(1)$.
    \item The probability of all good players outputting $0$ is $\Omega(1)$.
    \end{enumerate}
\end{lemma}

\begin{proof}
    (1) Procedure \textsc{Quantum-multicast} takes one round. Since there are at
    least $n-t$ good players, every good player will return from
    $\textsc{Quantum-get-core}$ in $O(1)$ rounds. Other steps in
    $\textsc{Quantum-common-coin}$ also take $O(1)$ rounds. Thus all good players 
    terminate in $O(1)$ rounds.

    (2) Since every player initially toss $1$ with probability~$1-1/n$, the probability that
    all good players toss $1$ is $(1-1/n)^n\geq 1/4$ for $n\geq 2$. Thus 
    \(
    \Pr[\text{all good players output }1]\geq 
    \Pr[\text{all good players initially toss }1]\geq 1/4.
    \)

    (3) Sort players by the time they terminate
    \textsc{Quantum-common-coin}. Let $a_i$ be the $i$-th player that terminates
    \textsc{Quantum-common-coin} and $V_i$ be the set of players whose coin is
    recorded by player $a_i$. By Lemma \ref{lem:core}, $|\cap_{i=1}^m V_i|\geq
    n/2$ where $m\geq n-t$ denote the number of good players. We strengthen the
    adversary's power as arbitrarily choosing sets $V_1, \ldots, V_m$ as long as
    the condition $|\cap_{i=1}^m V_i|\geq n/2$ is satisfied. For a set $S$ of players, We say \emph{$S$
    contains $0$} if there exists some player $j\in S$ such that measuring the
    player $j$'s quantum coin $\ket{c}$ (Line~\ref{line:coin}) will result in
    $0$. Before the first player $a_1$ measures its ancilla
    (Line~\ref{line:measure}), all coins are in uniform superposition, so the
    adversary's choice of $V_1$ will not depend on quantum coin values. When
    $a_i$ measures its ancilla (Line~\ref{line:measure}), whether $V_i$ contains
    $0$ becomes known to the full-information adversary,     
    so the adversary's choice of $V_i$ will depend on the knowledge of whether
    $V_1, \ldots, V_{i-1}$ contain $0$. 

    Let $C_j:=\cap_{i=1}^j V_i$ denote the core set of players $a_1,\ldots,a_j$.
    The first observation is that \(\Pr[\text{all good players output }0]\geq
    \Pr[C_m\text{ contains }0]\). Next, we claim that $\Pr[C_m\text{ contains
    }0]=\Omega(1)$: Note that after player $a_j$ terminates
    \textsc{Quantum-common-coin}, the adversary's only knowledge about $C_j$ is
    whether $C_j$ contains $0$. Conditioned on whether $C_j$ contains $0$, any
    two subsets $X, Y\subseteq C_j$ of the same size still have equal
    probability of containing $0$. 
    Thus among applicable adversaries that minimize $\Pr[C_{m}\text{ contains }0]$, there exists an adversary $\mathcal{A}$ whose strategy does not depend on whether $C_1,\ldots, C_{m-1}$ contain $0$, which further implies that $\mathcal{A}$'s choice of $C_{m}$ is independent of coin values.
    Combined with the fact that
    $|C_m|\geq n/2$, we have 
    \[\Pr[C_{m}\text{ contains
    }0]\geq 1-\prod_{p\in C_m} \Pr[\text{player }p\text{'s coin is }1]\geq
    1-\left(1-\frac{1}{n}\right)^{\frac{n}{2}}\geq 1-e^{-\frac{1}{2}}.\]
    
\end{proof}

Plugging the above quantum common-coin protocol into the classical BA protocol in
Algorithm 44 of \cite{10.5555/983102}, we get an $O(1)$-round quantum full-information BA
protocol. Trivially, the quantum protocol
have the same round and communication complexity. 

Finally, we analyze the local resource consumption of quantum protocol. The
local computation of \textsc{Quantum-common-coin} needs $O(n)$ space and
$O(n^2)$ time, which is the same as the classical protocol. The (qu)bit length of
each message is $O(n)$, which is the same as the bit length in the
classical protocol. In the quantum agreement protocol, we additionally require
players holding all qubits from previous invocations of
\textsc{Quantum-common-coin} before decision to avoid unwanted collapse of the
entanglement. However, since the protocol terminates in expected $O(1)$ rounds,
it does not change the overall complexity of the BA protocol.

\subsection{Byzantine case}
We quantize the classical protocol from \cite{bangalore2018almost}, which also
follows the framework of reducing BA protocols to common-coin protocols.
Their common-coin protocol is based on the design of a \emph{shunning
asynchronous verifiable secret sharing (SAVSS)} scheme. SAVSS is a two-phase
protocol (\textsc{Sh}, \textsc{Rec}) against an Byzantine adversary that can
corrupt up to $t$ players. In \textsc{Sh} phase, a designated player $D$ called
dealer shares a secret $s$ among players and players joint verify whether $s$ is
can be uniquely reconstructed later. If $D$ is good, then the secrecy of $s$ is
preserved until \textsc{Rec} phase. In \textsc{Rec} phase, all players jointly
reconstruct a unique secret $s$ or at least a certain number of bad players are
shunned by some good players. 

Since all other components of the protocol are deterministic, we only need to
quantize the SAVSS scheme. Also, the SAVSS scheme is only used to share a
uniformly random value. Thus it is sufficient to construct a quantum SAVSS
protocol sharing a uniformly random value in $\mathbb{F}_p$ where $p>n$ is a
prime number against full-information Byzantine adversary (\Cref{algo:avss}).
Note that in the algorithm, \emph{broadcast} refers to the reliable broadcast
primitive of Bracha \cite{bracha1987asynchronous}, and Lines~\ref{line:fol1} and \ref{line:fol2} refer to
the corresponding parts in Fig. 1 of \cite{bangalore2018almost}.

\begin{algorithm}[!htb]
    \caption{Quantum SAVSS protocol: code for player $i$.} \label{algo:avss}
    \begin{algorithmic}[1]
        \Procedure{\textsc{Quantum-sh}}{}
        \If {$i$ is dealer}
        \State prepare in new registers $\Dsf\Ssf_1\Ssf_2\cdots \Ssf_n$ a quantum state
        \[
        \quad \quad \sum_{\mathrm{symmetric} f\in \mathbb{F}_p^{\leq t}[x,y]} \ket{f(x,y)}_{\Dsf}\ket{f(x,1)}_{\Ssf_1}\ket{f(x,2)}_{\Ssf_2}\cdots\ket{f(x,n)}_{\Ssf_n}
        \]
        \For{$j\in [n]$}{ send $\Ssf_j$ to player $j$}\EndFor 
        \EndIf
        \Upon{$\Ssf_i$ is received from dealer} 
        \State prepare an empty register $\Qsf_i:=\Qsf_i^{(1)}\Qsf_i^{(2)}\cdots \Qsf_i^{(n)}$ 
        \State apply $U_{eval}$ on $\Ssf_i\Qsf_i$ where
        \[
        \quad \quad U_{eval}\ket{g(x)}_{\Ssf_i}\ket{0}_{\Qsf_i}:=
        \ket{g(x)}_{\Ssf_i}\ket{g(1)}_{\Qsf_i^{(1)}}\ket{g(2)}_{\Qsf_i^{(2)}}\cdots\ket{g(n)}_{\Qsf_i^{(n)}}
        \]
        \For{$j\in [n]$}
        \State prepare an empty register $\Rsf_i^{(j)}$
        \State apply $CX$ on $\Qsf_i^{(j)}\Rsf_i^{(j)}$ where 
        \(
            CX \ket{a}\ket{b} = \ket{a}\ket{a+b}
        \)
        \State send $\Rsf_i^{(j)}$ to player $j$
        \EndFor 
        \EndUpon
        \Upon{$\Rsf_{j}^{(i)}$ is received from player $j$}
        \State apply $CX^{-1}$ on $\Qsf_i^{(j)}\Rsf_j^{(i)}$ and measure $\Rsf_j^{(i)}$ \label{line:measureR}
        \If{measure result is $0$}
        {broadcast message ``player $j$ is ok''}
        \EndIf
        \EndUpon
        \State follow the original protocol to verify sets $\mathcal{V},\mathcal{V}_1,\ldots,\mathcal{V}_n$ and populate set $\mathcal{W}_{i,\mathrm{sid}}$ \label{line:fol1}
        \State {\bf return} $\Ssf_i$ as the share
        \EndProcedure
        \item[]
        \Procedure{\textsc{Quantum-rec}}{$\mathsf{S}_i$}
        \State measure $\Ssf_i$ and broadcast the measurement result \label{line:measureS}
        \State follow the original protocol \label{line:fol2}
        \EndProcedure
    \end{algorithmic}    
\end{algorithm}

We briefly explain why \Cref{algo:avss} satisfies the SAVSS properties
(Definition 2.1 of \cite{bangalore2018almost}). 
\begin{itemize}
\item {\bf Termination:} Follow directly from the original protocol. 
\item {\bf Correctness:} Let $f_i(x)$ be the polynomial obtained from measuring
$\Ssf_i$. Observe that if good player $i$ gets $0$ when measuring
$\Rsf_j^{(i)}$ (Line~\ref{line:measureR}), then $f_i(j)=f_j(i)$ must hold. Then the correctness follows
from the correctness of the original protocol. 
\item {\bf Privacy:} (i) Since good players will not measure the shares until
\textsc{Quantum-rec}, the adversary can operate on at most $t$ shares, which is
not enough to collapse the value of $f(0, 0)$. (ii) Measuring $\Ssf_i$ in
\textsc{Quantum-rec} phase (Line~\ref{line:measureS}) will not give full-information adversary more
knowledge since the player has already decided to broadcast the results.
\end{itemize}

Plugging the quantum SAVSS protocol into the original classical protocol, we get
a quantum asynchronous full-information BA protocol with the same round and communication complexity. Finally, we analyze
the local resource consumption. Since the quantum SAVSS protocol is obtained
simply by purifying all local computation, so it has the same local time/space
complexity and length of messages as the classical SAVSS protocol. Moreover, all
qubits are measured after one invocation of SAVSS, so no additional qubits need
to be preserved between invocations. Thus the quantum BA protocol has the
same complexity as the classical BA protocol. 

\end{document}